\documentclass[AMA,STIX1COL]{WileyNJD-v2}

\usepackage{graphicx}
\usepackage{mathtools,amsmath,amsfonts} 
\usepackage{psfrag}
\usepackage{caption,subcaption}
\usepackage{caption}
\usepackage{subcaption}
\usepackage{multirow}
\usepackage{epsfig,epstopdf}
\usepackage{algorithm}
\usepackage{enumitem}
\usepackage{hyperref}

\newcommand\calc{\ensuremath{\mathcal C}}
\newcommand\real{\ensuremath{\mathbb R}}

\newcommand\kset{\ensuremath{\mathbb K}}

\DeclareMathOperator{\argmin}{argmin}

\newtheorem{opt}{Optimization Problem}

\articletype{Article Type}%

\begin{document}

\title{{LMI} relaxations and its application to data-driven control design for switched affine systems
}

\author[1]{Alexandre Seuret*}

\author[1]{Carolina Albea}

\author[1]{Francesco Gordillo}

\authormark{A. Seuret \textsc{et al}}

\address[1]{\orgdiv{Dpto de Ingeniería de Sistemas y Automática}, \orgname{Univ. de Sevilla}, \orgaddress{\country{Spain}}}

\corres{*Alexandre Seuret. \email{aseuret@us.es}}

\presentaddress{Camino de los Descubrimientos, s/n 41092 Sevilla, Spain}

\abstract[Summary]{The problem of data-driven control is addressed here in the context of switched affine systems. This class of nonlinear systems is of particular importance when controlling many types of applications in electronic, biology, medicine, etc. Still in the view of practical applications, providing an accurate model for this class of systems can be a hard task, and it might be more relevant to work on data issued from some trajectories obtained from experiments and to deploy a new branch of tools to stabilize the systems that are compatible with the processed data. Following the recent concept of data-driven control design, this paper first presents a generic equivalence lemma that shows a matrix constraint based on data, instead of the system parameter. Then, following the concept of robust hybrid limit cycles for uncertain switched affine systems, robust model-based and then data-driven control laws are designed based on a Lyapunov approach. The proposed results are then illustrated and evaluated on an academic example.}


\maketitle
\footnotetext{\textbf{Abbreviations:} LMI, Linear Matrix inequality}

\section{Introduction }\label{sec1}

Over the past decades, robust control theory has been developed to solve analysis and control problems for a wide class of dynamical systems \cite{ebihara2015s,postlethwaite2007robust,scherer2001theory}. The main advantages of this method rely on the possibility to derive stability and stabilization tests as to evaluate and measure the robustness of systems as to control these uncertain systems. These tests are often written in the form of Linear Matrix Inequalities (LMIs) that can be easily and efficiently solved using semi-definite programming.  Generally, this direction of research has a common requirement, namely knowledge of a model, even though some parameters in the model might be subject to uncertainties modeled as norm-bounded or polytopic uncertainties. Recently, due to the emergence of artificial intelligence and the possibility of storing numerous data, the area of data-based identification has been enhanced.  The main motivation for this new paradigm is to avoid or limit the modeling phase to some extents and to better rely on a set of several experiments providing data in order to compensate for the lack of knowledge on the model. These methods have been developed to eliminate the common feature described above.  

In this direction, reinforcement learning methods have been developed in the literature; see, for instance, \cite{boczar2018finite},\cite{recht2019tour} to cite only a few of them. They mainly deal with system identification, estimation of the model, but rarely address closed-loop guarantees, and only few papers provide non-conservative constructive stabilization conditions for the closed-loop system using noisy data of finite length, which remains an open problem, even if the data are generated by a Linear Time-Invariant (LTI) system.
\newline
Indeed, the automatic control community has recently shown a growing interest in the problem of data-driven control design for various classical control problems (see, for instance \cite{berberich2020robust,berberich2020data,hou2013model}), enhancing the tools arising from robust control theory. In this context, the problem can be summarized as follows: \textit{How to translate a model-based stability or stabilization criteria into a data-driven ones, without introducing conservatism?} \newline A catalog of formulas has been provided recently in \cite{dePersis2019formulas}, where several problems related to the design of state/output feedback controllers for linear systems have been considered. It is notably shown therein that in the case of exact data experiments arising from linear time-invariant systems, i.e., without noise or uncertainties, equivalent formulations between model-based and data-based criteria have been obtained. More recently, another elegant solution to the data-driven design problem for linear systems has been provided in \cite{vanWaarde2020noisy,van2020data}. 

In this paper, an extension of the S-procedure has been proposed to eliminate in an LMI condition the uncertainties arising, for instance, from the model matrices. The underlying idea is to embed the uncertainties brought by the data experiments into an LMI constraint, which is considered an assumption in the design method. The authors have interestingly demonstrated that this assumption allows eliminating the model matrices at the price of introducing a single decision variable to the initial conditions for stability or stabilization. A similar assumption was also made in \cite{berberich2020robust,berberich2020combining}.
To the best of our knowledge, this trend of research has mainly been considered for linear systems. For the nonlinear cases, the reader may look at the case of bilinear input-affine systems \cite{bisoffi2020data}, referring to many relevant applications in engineering, medicine, or ecology \cite{mohler1973bilinear}. The authors demonstrate therein how their formulas provided in \cite{dePersis2019formulas} for the linear case can be efficiently adapted to this class of nonlinear systems.

In this paper, the objective is to demonstrate that data-driven methods can also be applied to a particular class of nonlinear systems, namely switched affine systems. This class of systems represents a highly relevant theoretical and practical area of research. From a practical point of view, they have been employed to model numerous applications such as embedded systems, electronic power converters, mixing fluids, damping of vibrating structures, mobile sensor networks \cite{antunes2016linear},  characterizing these complex systems and not intuitive behaviors. From a theoretical point of view, they represent a particular class of hybrid dynamical systems~\cite{goebel2012hybrid}\cite{liberzon2003switching}. A particularity of switched affine systems arises from the fact that, in general, it is not possible to stabilize their solution to a single equilibrium point but rather to a (hybrid) limit cycle, as shown in \cite{egidio2020global} and later refined in \cite{serieye2020stabilization}. 

 In the context of data-driven methods for switched systems (affine or not), the authors of \cite{kenanian2019data,rotulo2021online} provided the first attempts on switched systems (without affine terms). However, there are only a few works considering data-driven control design for switched affine systems, which is the objective of this paper. \textcolor{blue}{To the best of our knowledge, the only contribution in this direction has been presented in \cite{della2021data}, where the authors consider the stability analysis of this class of systems, when the switching signal is assumed to be an arbitrary exogenous input.  The authors of \cite{smarra2018data} use machine learning algorithms, especially regression trees and random forests to model a switched affine system using only historical data. The model is later used to establish a model predictive control or more precisely a Data Predictive Control (DPC) to obtain optimal control system trajectories. In this paper, the model under consideration is of the form $x^+=A_\sigma x +B_\sigma u+ F_\sigma$, where $x$ is the state (and $x^+$ the forward value of  $x$), $u$ the control law and $\sigma$ the switching signal. This system refers to a different class of switched affine systems since it is composed of a switched linear controlled systems $(A_\sigma x +B_\sigma u)$ with affine switching terms $(f_\sigma)$ where $\sigma$ is an exogenous signal.}
\newline
In contrast to these works, here we consider a data-driven control design method for switched affine systems, where the control variable is the switching signal. Following the framework provided in \cite{serieye2022attractor} dealing with an uncertain model-based method, this paper will present a robust data-driven control design for this class of systems. \textcolor{blue}{To do so the following objectives will be addressed in the paper
\begin{itemize}
\item Provide a technical generic tool that allows transforming a model-based criteria for discrete-time systems into a data-driven one.
\item Solve a problem of model-based design of a stabilizing switching state-feedback control law for switched affine systems subject to an external disturbance.
\item Illustrate the potential of the preliminary technical tool though the data-driven design of switching control laws for switched affine systems. 
\end{itemize}}

The structure of the paper is the following. Section 2 introduces a novel matrix-constraint relaxation, its relationship with the existing literature and a short discussion about its potential on the data-driven stabilization of linear systems. Then, Section 3 deals with the model-based stabilization of switched affine systems subject to a bounded external disturbance. Then, thanks to the results in the two previous sections, the main contribution of the paper is presented in Section 4 addressing the data-driven stabilization of switched affine systems. These results are then illustrated in Section 5, where a numerical application of the theoretical contributions is treated.

 \textbf{Notations:} Throughout the paper, $\mathbb N$ denotes the set of natural numbers, $\real$ the real numbers, $\real^n$ the $n$-dimensional Euclidean space, $\real^{n\times m}$ the set of all real $n \times  m$ matrices and $\mathbb S^n$ the set of symmetric matrices in $\real^{n\times n}$.  For all scalars $0<a<b$, notation $[a,b]_\mathbb N$ stands for $[a,b]\bigcap \mathbb N$, which represents the set of integers included in $[a,b]$. For any $n$ and $m$ in $\mathbb N$, matrices $I_n$ and $\mathbf{0}_{n,m}$ ($\mathbf{0}_{n}=\mathbf{0}_{n,n}$) denote the identity matrix of $\mathbb R^{n\times n}$ and the null matrix of $\mathbb R^{n\times m}$, respectively. For any integer $n$, $\mathbf{1}_{n}$ stands for the vector in $\mathbb R^n$, whose entries are all equal to $1$. When no confusion is possible, the subscripts of these matrices that precise the dimension will be omitted. For any matrix $M$ of $\mathbb R^{n\times n}$, the notation $M\succ0$, ($M\prec0$) means that $M$ is symmetric positive (negative) definite and $\det(M)$ represents its determinant. For any matrices $A=A^\top,B,C=C^\top$ of appropriate dimensions, matrix $\left[\begin{smallmatrix}A&B\\\ast & C  \end{smallmatrix} \right]$ denotes the symmetric matrix $\left[\begin{smallmatrix}A&B\\ B^\top& C  \end{smallmatrix} \right]$. $\Vert x \Vert$ denotes the Euclidean norm of $x$. 
 For a matrix $M\in\mathbb S^n$, $M\succ 0$ and a vector $h\in\mathbb R^{n}$, we denote the shifted ellipsoid $\mathcal E(M, h)=\left\{
x\in\mathbb R^n,~ (x-h)^\top M (x-h)\leq 1\right\}$.

\section{Matrix-constrained Relaxation} 
\label{Sec:Prel}

\subsection{New lemma for data-driven analysis}
The following lemma, which is the first contribution of the paper, presents a generic method to transform a problem of a particular matrix inequality depending on parameters that verify a quadratic constraint into a formulation that is independent of these parameters. Several solutions to this problem, also known as Quadratic Matrix Inequality, are stated below.

\begin{lemma}\label{lem0}
For given positive integers $n,m,q$, consider matrices $(\mathcal M_1,\mathcal M_2,\mathcal N_1,\mathcal N_2)$ in $\mathbb S^{q}\times  \mathbb S^{n}\times  \mathbb R^{q\times n}\times  \mathbb R^{q\times m}$ and a symmetric matrix $\Psi=\begin{bsmallmatrix}
 \Psi_1&\Psi_2\\
 \ast &\Psi_3\\
\end{bsmallmatrix}\in\mathbb S^{n+p}$  such that $\Psi_3\prec 0$. 

Then, the following statements are equivalent.
\begin{itemize}
\item[\textit{(i)}] Inequality 
\begin{equation}\label{lem0_1}
\mathcal M(\mathcal A)=\begin{bmatrix}
\mathcal M_{1} & \mathcal N_{1}+\mathcal N_2
\mathcal A^\top\\
\ast & \mathcal M_{2}
\end{bmatrix}\succ 0,\quad \forall \mathcal A\in  \Sigma(\Psi)
\end{equation}
holds. $ \Sigma(\Psi)$ represents the nonempty set of allowable uncertain matrices $\mathcal A$ characterized by a quadratic constraint defined as follows,
 \begin{equation}\label{def_Sigma}
 \Sigma(\Psi):=\left\{
 \mathcal A \in\mathbb R^{n\times m} \mbox{ s.t. } \left[\begin{matrix}
I_n\\
\mathcal A^\top\\
\end{matrix}\right]^\top  \begin{bmatrix}
 \Psi_1&\Psi_2\\
 \ast &\Psi_3\\
\end{bmatrix}\left[\begin{matrix}
I_n\\
\mathcal A^\top\\
\end{matrix}\right]\succeq 0
 \right\}.
 \end{equation}

\item[\textit{(ii)}] There exist matrices $(\mathcal R,\mathcal N) \in \mathbb S^{n}\times\mathbb  R^{q\times n}$  and a positive scalar $\eta>0$ such that 
\begin{equation}\label{lem:IneqY}
\begin{bmatrix}
\mathcal M_1 &
\mathcal N 
& \mathcal N_1- \mathcal N&\mathcal N_2\\
\ast &  \mathcal M_2-\mathcal R&0&0\\
\ast&\ast &  \mathcal R-\eta \Psi_1 & -\eta \Psi_2\\
\ast&\ast&\ast&-\eta \Psi_3\\
\end{bmatrix}\succ 0.
\end{equation}

\item[\textit{(iii)}] There exists a positive scalar $\eta>0$ such that 
\begin{equation}\label{lem:IneqY2}
\begin{bmatrix}
\mathcal M_1 &
\mathcal N_1 
&\mathcal N_2\\
\ast &  \mathcal M_2-\eta \Psi_1 & -\eta \Psi_2\\
\ast&\ast&-\eta \Psi_3\\
\end{bmatrix}\succ 0.
\end{equation}
\end{itemize}
\end{lemma}

\begin{proof} The proof is divided into three steps. 
\newline \underline{\textit{(i)$\Rightarrow$(ii):}} The first step of the proof is to find an appropriate expression of matrices $\mathcal A$ that belong to $\Sigma(\Psi)$. 
For any matrix $\mathcal A$ in $\Sigma(\Psi)$, it holds
$$\begin{array}{lcl}
0&\preceq& \left[\begin{matrix}
I_n\\
\mathcal A^\top\\
\end{matrix}\right]^\top \!\!\begin{bmatrix}
 \Psi_1&\Psi_2\\
 \ast &\Psi_3\\
\end{bmatrix}\left[\begin{matrix}
I_n\\
\mathcal A^\top\\ 
\end{matrix}\right]=
\tilde {\mathcal R}-\left[\begin{matrix}
I_n\\
\mathcal A^\top\\
\end{matrix}\right]^\top  \!\!\left[\begin{matrix} \tilde {\mathcal R}-\Psi_1& -\Psi_2\\ \ast &-\Psi_3 \end{matrix} \right]   \left[\begin{matrix}
I_n\\
\mathcal A^\top\\
\end{matrix}\right],
\end{array}
$$
where $\tilde {\mathcal R}$ is any matrix in $\mathbb S^n$. In addition, condition $\Psi_3\prec 0$ ensures that there exists a matrix $\tilde {\mathcal R}$ such that $ \left[\begin{smallmatrix}  \tilde {\mathcal R}-\Psi_1& -\Psi_2\\ \ast &-\Psi_3 \end{smallmatrix} \right] $ is positive definite, which allows applying the Schur complement as follows
$$
\begin{bmatrix}
\left[\begin{matrix} \tilde {\mathcal R}-\Psi_1& -\Psi_2\\ \ast &-\Psi_3 \end{matrix} \right] ^{-1} & \
\left[\begin{matrix}
I_n\\
\mathcal A^\top\\ 
\end{matrix}\right]\\
\ast& \tilde {\mathcal R}
\end{bmatrix}\succ0.
$$

Note that it is not the usual way to apply the Schur complement, but this dual way has been considered to keep block $(1,2)$ (resp. $(2,1)$) with $ \left[\begin{smallmatrix}
I_n\\
\mathcal A^\top\\ 
\end{smallmatrix}\right]$ (resp. its transpose). Next, pre- and post-multiply the previous inequality by $\left[\begin{smallmatrix}
\tilde {\mathcal N} &0\\ 0&I \end{smallmatrix}\right]$ and its transpose, respectively, where $\tilde {\mathcal N}$ is any matrix in $\mathbb R^{p\times (n+m)}$. Then, having $\mathcal A\in  \Sigma(\Psi)$ implies 
\begin{equation}\label{Lem0_Proof_cond}
\begin{bmatrix}
\tilde {\mathcal N} \left[\begin{matrix} \tilde {\mathcal R}-\Psi_1& -\Psi_2\\ \ast &-\Psi_3 \end{matrix} \right] ^{-1} \!\!\! \!\! \tilde {\mathcal N}^\top & \tilde {\mathcal N}\left[\begin{matrix}
I_n\\
\mathcal A^\top\\
\end{matrix}\right]\\
\ast& \tilde {\mathcal R}
\end{bmatrix}\succ0.
\end{equation}

The previous calculations ensure that inequality \eqref{lem0_1} can be rewritten as $\mathcal M(\mathcal A)\succ 0$, for all $\mathcal A$ so that \eqref{Lem0_Proof_cond} holds. Using an S-procedure, this statement is equivalent to the existence of a positive scalar $\eta>0$ such that
$$
\begin{bmatrix}
\mathcal M_{1} & \mathcal N_{1}\!+\!\mathcal N_2
\mathcal A^\top\\
\ast & \mathcal M_{3}
\end{bmatrix}\!-\!\eta 
\begin{bmatrix}
\tilde {\mathcal N} \left[\begin{matrix} \tilde {\mathcal R}\!-\!\Psi_1&\!\! -\Psi_2\\ \ast &\!\!-\Psi_3 \end{matrix} \right] ^{-1}
 \tilde {\mathcal N}^\top & \tilde {\mathcal N}\left[\begin{matrix}
I_n\\
\mathcal A^\top\\
\end{matrix}\right]\\
\ast& \tilde {\mathcal R}
\end{bmatrix}\succ 0,
$$
which, together with the Schur complement, writes
$$
\begin{bmatrix}
\mathcal M_{1} &\quad \left(\begin{bmatrix}\mathcal N_{1}&\mathcal N_{2}\end{bmatrix}-\eta \tilde {\mathcal N}\right)\left[\begin{matrix}
I_n\\
\mathcal A^\top\\
\end{matrix}\right] &\eta \tilde {\mathcal N}\\
\ast & \mathcal M_{2}-\eta \tilde {\mathcal R}& 0\\
\ast&\ast & \begin{bmatrix}\eta \tilde {\mathcal R}-\eta\Psi_1& -\eta\Psi_2\\ \ast &-\eta\Psi_3 \end{bmatrix} 
\end{bmatrix}
\succ 0.
$$

Selecting $\eta \tilde {\mathcal R}= {\mathcal R}$ and $\eta \tilde {\mathcal N} = \begin{bmatrix}\mathcal N_{1}-\mathcal N&\mathcal N_{2}\end{bmatrix} $ so that
$$
 \left(\begin{bmatrix}\mathcal N_{1}&\mathcal N_{2}\end{bmatrix}-\eta \tilde {\mathcal N}\right)\left[\begin{matrix}
I_n\\
\mathcal A^\top\\
\end{matrix}\right]=\begin{bmatrix} \mathcal N&0\end{bmatrix}\left[\begin{matrix}
I_n\\
\mathcal A^\top\\
\end{matrix}\right]=\mathcal N,
$$
makes disappear the dependence of the condition on the uncertain matrix $\mathcal A$ and concludes the first part of the proof.
\newline \underline{\textit{(ii)$\Rightarrow$(iii):}} Pre- and post-multiplying inequality \eqref{lem:IneqY} by $\begin{bsmallmatrix}I_n&0& 0&0\\ 
0& I_n & I_n &0\\
0&0&0&I_n
 \end{bsmallmatrix}$ and its transpose, respectively, leads to \eqref{lem:IneqY2}, where matrices $\mathcal R$ and $\mathcal N$ have been eliminated.
\newline \underline{\textit{(iii)$\Rightarrow$(i):}} Pre- and post-multiplying inequality \eqref{lem:IneqY2} by $\begin{bsmallmatrix}I_n&0& 0\\ 
0& I_n &   \mathcal A \end{bsmallmatrix}$ and its transpose, respectively, leads to 
$$
0\prec \begin{bmatrix}
\mathcal M_{1} &\mathcal N_{1}\!+\!\mathcal N_2
\mathcal A^\top\\
\ast & \mathcal M_{2}
\end{bmatrix}-\eta \begin{bmatrix}
0& 0\\
\ast &   \left[\begin{matrix}
I_n\\
\mathcal A^\top\\
\end{matrix}\right]^\top \!\!\! \begin{bmatrix} \Psi_1&\Psi_2\\
 \ast &\Psi_3\\
\end{bmatrix}
 \left[\begin{matrix}
I_n\\
\mathcal A^\top\\
\end{matrix}\right]
\end{bmatrix}.
$$

This inequality proves that inequality \eqref{lem0_1} holds for all matrices $\mathcal A$ that belong to $\Sigma(\Psi)$, since the second term is negative semi-definite.
\end{proof}

\subsection{Comparison with the S-Lemma} 
 
Lemma \ref{lem0} provides an alternative formulation in robust analysis for uncertain matrices subject to quadratic constraints of the form \eqref{def_Sigma} compared to the one presented in \cite{vanWaarde2020noisy,van2020data}. For the sake of consistency, the S-Lemma provided in \cite{vanWaarde2020noisy} is recalled in the following lemma.
\begin{lemma}{\cite[Th.9]{vanWaarde2020noisy}}\label{lem:Waarde}
Let $\mathcal M_s, \Psi_s \in \mathbb R^{(n_s+m_s)\times(n_s+m_s)}$ be symmetric matrices and assume that there exists some matrix $\mathcal A_s \in \mathbb R^{n_s\times m_s}$ such that
$
\begin{bsmallmatrix}
I_{n_s}\\
\mathcal A_s^\top
\end{bsmallmatrix}^\top  \Psi_s \begin{bsmallmatrix}
I_{n_s}\\
\mathcal A_s^\top
\end{bsmallmatrix}\succ0
$. Then the following statements are equivalent:
\begin{enumerate}
\item[(i)] $\begin{bmatrix}
I_{n_s}\\
\mathcal A_s^\top
\end{bmatrix}^\top \mathcal M_s \begin{bmatrix}
I_{n_s}\\
\mathcal A_s^\top
\end{bmatrix}\succ0,\ \forall \mathcal A_s\in \Sigma(\Psi_s) \subset \mathbb R^{n_s\times m_s}$, \newline where set $\Sigma(\Psi_s)$ has the same definition as in \eqref{def_Sigma} but replacing  $\Psi$ by $\Psi_s$.
\item [(ii)] There exists a scalar $\eta > 0$ such that $\mathcal M_s - \eta  \Psi_s \succ 0$.
\end{enumerate}
\end{lemma}

The main similarities and differences with respect to this formulation are described here. Let us first note that both lemmas address the same problem, consisting of the satisfaction of a matrix inequality subject to uncertain matrices characterized by a quadratic constraint. The main interest of both lemmas is to derive equivalent inequalities that are independent of the uncertain matrix $\mathcal A$ (or $\mathcal A_s$). Finally, both lemmas can be seen as application of the usual manipulations on LMI such as Schur Complement, Finsler's lemma and S-procedure. 
\newline
Apart from presenting these similarities, both lemmas have substantial differences. First, Lemma \ref{lem:Waarde} requires that matrix $\mathcal M_s$ has the same size as $\Psi_s$,  that characterizes the quadratic constraint on $\mathcal A$. Lemma~\ref{lem0} is more flexible in this sense, as there is no relationship between matrices $\mathcal M_1$ and $\Psi$, which are independent. This flexibility has the benefit of reducing the initial manipulations to derive, from usual stability or control problems, the appropriate expression of $\mathcal M$ and $\Psi$ to fit the framework of Lemma~\ref{lem:Waarde}. In fact, the relationship between both lemmas can be seen by selecting 
$$
\begin{array}{c}
n_s=q+n ,\quad m_s= m,\quad \mathcal A_s^\top =\begin{bmatrix} 0_{m\times q}& \mathcal A^\top \end{bmatrix},\quad
\mathcal M_s=\begin{bmatrix}
\mathcal M_1 &
\mathcal N_1 
&\mathcal N_2\\
\ast &  \mathcal M_2 &0\\
\ast&\ast&0\\
\end{bmatrix}, \quad \Psi_s= \begin{bmatrix}
0&
0
&0\\
\ast &\ \ \Psi_1 &  \Psi_2\\
\ast&\ast& \Psi_3\\
\end{bmatrix}.
\end{array}
$$

From this selection, it is clear that item (ii) of Lemma \ref{lem:Waarde} is equivalent to item (iii) of Lemma~\ref{lem0}, showing that Lemma~\ref{lem0} is a particular case of Lemma~\ref{lem:Waarde}.
Therefore, it is important to understand the advantages of Lemma~\ref{lem0} with respect to Lemma \ref{lem:Waarde}, since they can be used for the same purpose. The first one is related to the structure of matrix $\mathcal M(\mathcal A)$, which appears in many constructive LMI problems for the stabilization of discrete-time systems. As a first illustration, consider the problem of stabilization of a discrete-time linear system $x^+=Ax+Bu$ using a state-feedback controller $u=Kx$, where we adopted the following notation $x^+ = x_{k+1}$ and $x = x_k$. Using a quadratic Lyapunov function, with a positive definite matrix $P$, the condition for the stabilization of this system writes
$$
(A+BK)^\top P(A+BK)-P\leq0.
$$

Then, as usual in this context, the design of the control gain is obtained by introducing $W=P^{-1}$ and applying the Schur complement to the first term of the previous inequality, leading to the equivalent problem
$$
0\prec\mathcal M(A,B)= \begin{bmatrix}
W\ \ &WA^\top\!\!\! +\!WK^\top B^\top\\
\ast& W
\end{bmatrix}= 
\begin{bmatrix}
W&\begin{bmatrix}W& WK^\top \end{bmatrix} \begin{bmatrix}A^\top \\B^\top\end{bmatrix} \\
\ast& W
\end{bmatrix}.
$$

This usual manipulation drives naturally to the formulation of item (i) in Lemma~\ref{lem0} and, it is easy to identify that for this very simple example
$$
\mathcal M_1\!=\!\mathcal M_2\!=\! W, \quad \mathcal N_1\!=\! 0, \quad  \mathcal N_2\!=\!\begin{bmatrix}
W& WK^\top \end{bmatrix},\quad    \mathcal A \!=\! \left[\begin{matrix}A & B \end{matrix}\right].
$$

Then Lemma~\ref{lem0} leads to the following equivalent formulation
$$
0\prec
\begin{bmatrix}
W\ \ & 0 & \begin{bmatrix}W& WK^\top \end{bmatrix} \\
\ast\ \ & W-\eta \Psi_1\ \ &-\eta \Psi_2\\
\ast&\ast&-\eta \Psi_3
\end{bmatrix},
$$
which can be rewritten as an LMI by introducing a new decision variable $Y=KW$. 

Keeping the same example and following the development presented in \cite{vanWaarde2020noisy} or \cite{berberich2020robust}, the LMI stabilization problem cannot be treated directly. First, a manipulation is required to solve the problem, i.e. to consider the dual or transpose problem that is 
$$
(A+BK) P(A+BK)^\top-P\leq0.
$$

Then, it is possible to apply the equivalence formulation proposed in Lemma~\ref{lem:Waarde}. This manipulation is the key step in the developments provided in \cite{vanWaarde2020noisy}. Note that this manipulation is correct when considering this simple linear time-invariant example, since both LMIs are necessary and sufficient conditions for matrix $A+BK$ to be Schur stable and consequently for the transpose matrix $(A+BK)^\top$ to be Schur stable as well. However, this manipulation is not permitted any more, or at least has to be studied carefully when other classes of systems are considered such as systems subject to nonlinearities, saturations, etc.

In order to avoid working on the transpose matrix, one may apply the Schur complement and work on the same inequality $\mathcal M_s(A,B)=\mathcal M(A,B)$. However, by doing this, the problem fits exactly to the structure of the inequality presented in item (i) of Lemma~\ref{lem0}, for which no additional manipulations are further needed.

Even though this example may seem too simple, a more complicated stabilization problem will be considered in the next section on switched affine systems. In this situation, one can better understand the potential and simplicity of the proposed formulation. 
 This is the main motivation to use Lemma~\ref{lem0}.

Another important issue is related to item (ii) in Lemma~\ref{lem0}, which introduces two slack variables, namely $\mathcal R$ and $\mathcal N$. In light of item (iii), these slack variables are not needed, in general. However, there exist at least two situations in which these additional degrees of freedom could be useful. The first refers to the case where matrices in $\mathcal M(\mathcal A)$ are subject to uncertainties. A second case of interest appears in the context of an optimization problem, where these slack variables may ease the search for the optimal solution, as we will show in the latter example section.

\section{Robust hybrid cycles for perturbed switched affine systems}

In this section, the objective is to present constructive stabilization conditions for switched affine systems subject to a bounded external disturbance. After formulating the problem under consideration, several preliminaries on robust hybrid cycles and cycles for this class of systems will be documented. Then, the first contribution of this paper on the robust model-based stabilization of switched affine systems will be presented.

\subsection{System data}
\label{sec:prob}

Consider the discrete-time switched affine system governed by the following dynamics.
\begin{align}\label{eq:model_x}
	\left\{
		\begin{array}{l} 
			x^+ = A_{\sigma} x + B_{\sigma}+w,\\
			\sigma \in u(x) \subset \mathbb K,\\
			x_0\in \mathbb R^n,
		\end{array}
	\right.
\end{align}
where $ x \in \real^n$ is the state vector. \textcolor{blue}{At any time instant $k\in\mathbb N$, $x$ and $x^+$ stand for $x(k)$ and $x(k+1)$, respectively. In \eqref{eq:model_x}, the time argument $k\in\mathbb N$ is omitted for the sake of simplicity.}
Likewise, $\sigma \in \kset:=\{1,2,..,K\}$ characterizes the active mode. The dynamics of the system is affected by an external disturbance vector $w \in\mathbb R^n$ that verifies
\begin{equation}\label{eq:cond_w}
    w^\top w\leq \lambda^2,
\end{equation}
for some given positive real number $\lambda$. Finally, $A_j\in \mathbb R^{n\times n}$ and $B_j \in \mathbb R^{n\times 1}$ are the matrices of mode $j\in\mathbb K$, which are not necessarily constant and known but it is assumed that they possibly belong to a polytopic set of uncertainties. The particularity of this class of systems relies on its control action, which is only performed by selecting the active mode $\sigma$, which requires particular attention.

The objective here is to design a suitable set-valued map $u$ in system~\eqref{eq:model_x} that ensures the convergence of the state trajectories to a set to be characterized accurately. {\color{blue} Note that the property of $u$ of being a set-value map comes from the fact that $u\in\mathbb{K}$.} A model-driven solution to this problem was provided in \cite{serieye2022attractor} and will be summarized hereafter. Interestingly, this paper provides a solution for robust stabilization of switched affine systems, which consists of proving that the solutions to the closed-loop system converge to a robust limit cycle which is composed of the union of shifted ellipsoid. This solution will be recalled in the next section. 

This formulation paves the way for the problem of the data-driven design of stabilizing controllers for switched affine systems. Let us first formulate the problem of data-driven control.

\subsection{Cycles and robust limit cycles}

Following the ideas developed in \cite{serieye2022attractor}, the notion of limit cycles~\cite{Strogatz_Book_1994,Sun_CSF_2008} is adapted to the problem under consideration. Before presenting the concept of robust limit cycles, let us first introduce the following definitions.
%
\begin{definition}[Cycle]  
A cycle, denoted as $\nu$, is a periodic function from $\mathbb N$ to $\mathbb K$. More precisely, this means that there exists $N$ in $\mathbb N\backslash\{0\}$, such that 
$$\nu(i+N)= \nu(i),\quad \forall i \in \mathbb N.$$

For any cycle $\nu$, notations $N_\nu$ and $\mathbb D_\nu$ stand for the minimum period and the minimum domain of $\nu$, respectively. More formally, they are defined as follows
$$
\begin{array}{rcl}
			N_\nu&:= &\min {N\in \mathbb N\backslash\{0\}}  \mbox{ s.t. }\nu(i+N)= \nu(i) \in \mathbb K,\quad  \forall i \in \mathbb N,\\
			\mathbb D_\nu&:=&\{1,2,\dots, N_\nu\}.
\end{array}
$$
\end{definition}
\begin{definition}[Set of cycles]
	Denote the set of cycles from $\mathbb N$ to $\mathbb K$ by 
	\begin{equation*}\label{eq:def_cycle}
		\calc:=\left\{\nu : \mathbb {N} \rightarrow \kset,\ \mbox{s.t.}\ \exists N \in \mathbb N\backslash\{0\},\quad  \forall i\in \mathbb {N},\ \nu(i+N)=\nu(i)
		\right\}.
	\end{equation*}	
\end{definition}
To ease readability, we introduce the following modulo notation: $\lfloor i\rfloor_{\nu}=((i-1) \textrm{ mod }  N_\nu)+1$, for any $i\in\mathbb N\backslash \{0\}$. In particular, $\lfloor i\rfloor_{\nu}=i$, for any $i=1,\dots,N_\nu$ and $\lfloor N_\nu+1\rfloor_{\nu}=1$. The notion of robust limit cycles, which extends the definition of limit cycles in~\cite{Strogatz_Book_1994,Sun_CSF_2008}, is adapted to discrete-time systems and is defined below. 
\begin{definition}[Robust Limit Cycles]
System \eqref{eq:model_x}  admits a robust limit cycle associated with a cycle $\nu\in\mathcal C$ if there exist possibly disjoint subsets $\mathcal L_i \subset \mathbb R^n$,  for $i\in\mathbb D_\nu$ such that 
\begin{equation}\label{Def:robustLimitCycle}
		A_{\nu(i)} \mathcal L_i + B_{\nu(i)} \subset  \mathcal L_{\lfloor i\!+\!1\rfloor_\nu}, \ \ \forall i \in \mathbb D_\nu.
\end{equation}
\end{definition}
In inclusion \eqref{Def:robustLimitCycle}, the left-hand side of the inclusion means, with a light abuse of notation, that, for any $ i\in\mathbb D_\nu$ and for all $x\in \mathcal L_i$,  vector $A_{\nu(i)} x+B_{\nu(i)}$ belongs to $\mathcal L_{\lfloor i+1\rfloor_\nu}$. In the sequel,  Figure~\ref{fig:Rob_Lim_Cycle} illustrates this set of inclusions and the relationship with the cycle.
\begin{figure}
		\centering
		       \hspace{-0.15cm} \includegraphics[width=0.25\textwidth]{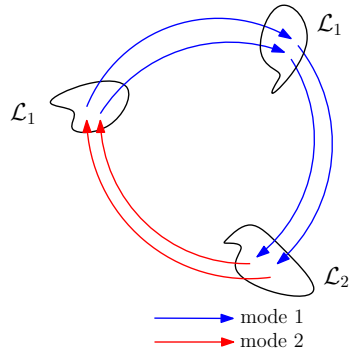}
	\caption{Illustration of a robust limit cycle composed of three subsets and associated to the cycle $\nu=\{1,1,2\}$. Inclusions \eqref{Def:robustLimitCycle} impose the invariance of the trajectories inside the robust limit cycle.}
		\label{fig:Rob_Lim_Cycle}
	\end{figure}
If, in some cases, the subsets composing the robust limit cycles are reduced to singletons, i.e. $\mathcal L_i=\{\rho_i\}$, for all $i\in\mathbb D_\nu$, then inclusion \eqref{Def:robustLimitCycle} is rewritten as a set of equalities given by
\begin{equation}\label{def:cycleq}
		\rho_{\lfloor i+1\rfloor_{\nu}}=A_{\nu(i)} \rho_{i} + B_{\nu(i)},\quad \forall i\in\mathbb D_\nu,
\end{equation}
illustrating then that inclusions \eqref{Def:robustLimitCycle} are the natural extension of \eqref{def:cycleq}.

Note that the idea of studying limit cycles for switched affine systems has been considered in \cite{egidio2020global} for the case of constant and known matrices $A_j$ and $B_j$. In this case, necessary and sufficient constructive conditions for the existence of $\{\rho_i\}_{i\in\mathbb D_\nu}$, for a given cycle $\nu$, have been provided in \cite{serieye2022attractor}, while it was only an assumption in \cite{egidio2020global}. It is also worth noting that the structure of the control law provided in \cite{egidio2020global} complicates the study of robust limit cycles since their control law requires exact knowledge of the system matrices.

\subsection{Robust model-based stabilization of switched affine systems}

A solution to the stabilization of switched affine systems subject to external disturbances to a robust limit cycle is presented here following the idea and concept borrowed from \cite{serieye2022attractor}. The main difference with respect to \cite{serieye2022attractor} is the addition of the bounded external disturbance in \eqref{eq:model_x}. The robust stabilization problem is formalized as follows.

\begin{theorem}\label{th:theo_robust0}
	For a given cycle $\nu$ in $\calc$ and for a parameter $\mu \in \left(0,1\right)$, consider the solution $\{(W_i,\zeta_i,\delta_i)\}_{i\in \mathbb D_\nu}$ in $\mathbb{S}^{n}\times \mathbb R^{n}\times \mathbb R$ to the following problem. 
	\begin{equation}\label{eq:th_LMI_rob}
	\begin{array}{rcl}
	\Phi_i (A_{\nu(i)},B_{\nu(i)}) \succ0, \quad W_i\succ 0, \quad \delta_i >0, \quad \forall i\in \mathbb D_\nu,
	\end{array}
	\end{equation}
	where matrices $\Phi_i$ are defined for all $i$ in $\mathbb D_{\nu}$ as follows.
	\begin{equation}\label{eq:WiPsi0}
	\begin{array}{lclllcl}
		\Phi_i (A_{\nu(i)},\!B_{\nu(i)}) &\!=&
		\!\begin{bmatrix}
		    	(1 \!-\! \mu)W_i &\!0  & 0& \!W_i A_{\nu(i)}^\top \\
		    	\ast &\! \mu \!-\!\delta_i \lambda^2&\!0 &\zeta_i^\top A_{\nu(i)}^\top \!+\!B_{\nu(i)}^\top\!-\!\zeta_{\lfloor i+1\rfloor_\nu}^\top \\    
		    	\ast &\!\ast &\! \delta_i  I&I\\
		    	\ast &\! \ast &\!\ast &\! W_{\lfloor i+1\rfloor_\nu} 
		    \end{bmatrix}
		   \end{array}.
	\end{equation}
	Then, attractor 	\begin{equation}\label{def:S0}
			\mathcal{S}_\nu:=\bigcup_{i\in\mathbb D_\nu}\mathcal E( W_i^{-1}, \zeta_i) 
		\end{equation}
		with $$\mathcal E( W_i^{-1}, \zeta_i)=\left\lbrace x \in \mathbb{R}^n, i\in\mathbb D_\nu \big| (x-\zeta_i)^\top W^{-1}_i(x-\zeta_i) \leq 1 \right\rbrace,$$
		is robustly globally exponentially stable for system~\eqref{eq:model_x} with the disturbance signal \eqref{eq:cond_w} and with
		the switching control law
		\begin{equation}\label{eq:control_law_robust0}
			u(x) = \left\{ \nu\left( \theta \right),~\theta\in\underset{i\in\mathbb D_\nu}{\argmin}\left(x-\zeta_i\right)^\top W_i^{-1}\left(x-\zeta_i\right) \right\}\subset \mathbb K.
		\end{equation}
\end{theorem}

\begin{proof} The proof of this theorem is largely inspired by \cite{serieye2022attractor}. This theorem is demonstrated thanks to the Lyapunov function built with matrices $W_i^{-1}\succ0$ and vectors $\zeta_i$'s that are the decision variables of  \eqref{eq:WiPsi0} and given by 
\begin{equation}
V(x) =  \min_{i\in\mathbb D_\nu}(x-\zeta_i) ^\top W^{-1}_i(x-\zeta_i).
\end{equation}	

The forward increment of the Lyapunov function writes
\begin{equation}
\begin{array}{lcl}
\Delta V(x) &=& \displaystyle \min_{i\in\mathbb D_\nu}(x^+\!-\!\zeta_i) ^\top W^{-1}_i(x^+\!-\!\zeta_i)
\!-\!\min_{\iota\in\mathbb D_\nu}(x\!-\!\zeta_\iota) ^\top W^{-1}_\iota(x\!-\!\zeta_\iota)\\
&=& \displaystyle \min_{i\in\mathbb D_\nu}(x^+\!-\!\zeta_i) ^\top W^{-1}_i(x^+\!-\!\zeta_i)\!-\!(x\!-\!\zeta_\theta) ^\top W^{-1}_\theta(x\!-\!\zeta_\theta).
\end{array}
\end{equation}	

The last equation holds since $\theta$ is the value of $\mathbb D_\nu$ that minimizes $(x\!-\!\zeta_\iota) ^\top W^{-1}_\iota(x\!-\!\zeta_\iota)$ according to the control law \eqref{eq:control_law_robust0}. Furthermore, we have
$$
\min_{i\in\mathbb D_\nu}(x^+\!-\!\zeta_i) ^\top W^{-1}_i(x^+\!-\!\zeta_i)\leq (x^+\!-\!\zeta_{\lfloor \theta+1\rfloor_\nu} ) ^\top W^{-1}_{\lfloor \theta+1\rfloor_\nu}(x^+\!-\!\zeta_{\lfloor \theta+1\rfloor_\nu})
$$
since the minimum is always lower than or equal to any other term. This expression only considers a particular case, which is $i={\lfloor \theta+1\rfloor_\nu}$. All together, an upper bound of the increment of the Lyapunov function is derived as follows
\begin{equation}
\begin{array}{lcl}
\Delta V(x) \!\leq \! \displaystyle (x^+\!\!-\!\zeta_{\lfloor \theta\!+\!1\rfloor_\nu}\! ) ^{\!\top} W^{-1}_{\lfloor \theta\!+\!1\rfloor_\nu}\!(x^+\!\!-\!\zeta_{\lfloor \theta+1\rfloor_\nu}\!)\!-\!(x\!-\!\zeta_\theta) ^{\!\top} W^{-1}_\theta(x\!-\!\zeta_\theta).
\end{array}
\end{equation}	

Next, we note that
$$
x^+-\zeta_{\lfloor \theta+1\rfloor_\nu}
			= A_{\nu(\theta)} (x-\zeta_\theta) +  \underbrace{A_{\nu(\theta)} \zeta_\theta+ B_{\nu(\theta)} - \zeta_{\lfloor \theta+1\rfloor_\nu}}_{\mathcal{B}_{\nu\left(\theta\right)}}+w ,
$$
where $ \mathcal{B}_{\nu\left(\theta\right) }$ is not necessarily zero. Let us now introduce $\chi_\theta^\top = [ (W^{-1}_\theta(x - \zeta_{\theta}))^\top \ 1\ \omega^\top ]$, so that 
$$
\begin{array}{lcl}
x - \zeta_{\theta} &=& [ W_\theta  \ \ \ 0\ \ \ 0]\chi_\theta\\
x^+ - \zeta_{\lfloor \theta+1\rfloor_\nu} &=& [ A_{\nu\left(\theta\right)} W_\theta  \ \ \ \mathcal B_{\nu(\theta)}\ \ \ I]\chi_\theta.
\end{array}
$$

Using this notation, the increment of the Lyapunov function can be expressed as follows:
\begin{equation}
\begin{array}{lcl}
\Delta V(x) \!\leq \! \displaystyle \chi_\theta^\top\left(
\begin{bsmallmatrix} W_\theta A_{\nu\left(\theta\right)} ^\top  \\ \mathcal B_{\nu(\theta)}^\top\\ I\end{bsmallmatrix}  W^{-1}_{\lfloor \theta\!+\!1\rfloor_\nu}
\begin{bsmallmatrix} W_\theta A_{\nu\left(\theta\right)} ^\top  \\ \mathcal B_{\nu(\theta)}^\top\\ I\end{bsmallmatrix}^\top 
-
\begin{bsmallmatrix} W_{ \theta}  \\ 0\\ 0\end{bsmallmatrix}  W_{ \theta}^{-1}\begin{bsmallmatrix} W_{ \theta}  \\ 0\\ 0\end{bsmallmatrix}^\top\right)\chi_\theta .
\\
\end{array}
\end{equation}

In addition to the previous inequality, we need to include the constraint on the external disturbance \eqref{eq:cond_w} as well as the fact that we will require convergence to the attractor $\mathcal S_\nu$, i.e., for all $x\in \mathbb R^n$ such that $V(x)\leq 1$. Using two S-procedures, this means that there exist two positive scalars $\mu>0$ and $\delta>0$, such that 
\begin{equation}
\begin{array}{lcl}
\Delta V(x) &\leq& \displaystyle \chi_\theta^\top\left(
\begin{bsmallmatrix} W_\theta A_{\nu\left(\theta\right)} ^\top  \\ \mathcal B_{\nu(\theta)}^\top\\ I\end{bsmallmatrix}  W^{-1}_{\lfloor \theta\!+\!1\rfloor_\nu}
\begin{bsmallmatrix} W_\theta A_{\nu\left(\theta\right)} ^\top  \\ \mathcal B_{\nu(\theta)}^\top\\ I\end{bsmallmatrix}^\top 
-\begin{bsmallmatrix}  W_{ \theta}  \\ 0\\ 0\end{bsmallmatrix}  W_{ \theta}^{-1}\begin{bsmallmatrix} W_{ \theta}  \\ 0\\ 0\end{bsmallmatrix}^\top\right)\chi_\theta
-\mu(V(x)-1)-\delta(\lambda^2-w^\top w)\\
&=& \displaystyle \chi_\theta^\top\left(
\begin{bsmallmatrix} W_\theta A_{\nu\left(\theta\right)} ^\top  \\ \mathcal B_{\nu(\theta)}^\top\\ I\end{bsmallmatrix}  W^{-1}_{\lfloor \theta\!+\!1\rfloor_\nu}
\begin{bsmallmatrix} W_\theta A_{\nu\left(\theta\right)} ^\top  \\ \mathcal B_{\nu(\theta)}^\top\\ I\end{bsmallmatrix}^\top 
-\begin{bsmallmatrix}  (1-\mu) W_{ \theta}  &0 &0 \\ \ast & \mu-\delta\lambda^2& 0\\ \ast&\ast & \delta I\end{bsmallmatrix}  \right)\chi_\theta
\end{array}
\end{equation}

Applying the Schur complement leads to the expression of $\Phi_i(A_{\nu(i)},B_{\nu(i)})$.
If all matrices $\Phi_i(A_{\nu(i)},B_{\nu(i)})$ in \eqref{eq:WiPsi0} are positive definite, then the increment of the Lyapunov function is negative definite outside of $\mathcal S_\nu$.

\textcolor{blue}{To show that $\mathcal S_\nu$ is also an invariant set of the closed-loop system \eqref{eq:model_x}, it suffices to write  $V(x^+)=V(x)+\Delta\! V(x) $.  Then, enforcing the introduction of the terms employed in the S-procedure during the first step of the proof, we get
$$
\begin{array}{lcl}
V(x^+)&\!=&\!V(x)-( \mu(V(x)\!-\!1)+\delta(\lambda^2-w^\top w) )+ \Delta\! V(x) +( \mu(V(x)\!-\!1)+\delta(\lambda^2-w^\top w) )
\end{array}
$$}

\textcolor{blue}{Following the previous developments, the previous expression writes
$$
\begin{array}{lcl}
V(x^+)&\!=&\!V(x)-( \mu(V(x)\!-\!1)+\delta(\lambda^2-w^\top w) )
\displaystyle +\chi_\theta^\top\left(
\begin{bsmallmatrix} W_\theta A_{\nu\left(\theta\right)} ^\top  \\ \mathcal B_{\nu(\theta)}^\top\\ I\end{bsmallmatrix}  W^{-1}_{\lfloor \theta\!+\!1\rfloor_\nu}
\begin{bsmallmatrix} W_\theta A_{\nu\left(\theta\right)} ^\top  \\ \mathcal B_{\nu(\theta)}^\top\\ I\end{bsmallmatrix}^\top 
-\begin{bsmallmatrix}  (1-\mu) W_{ \theta}  &0 &0 \\ \ast & \mu-\delta\lambda^2& 0\\ \ast&\ast & \delta I\end{bsmallmatrix}  \right)\chi_\theta
\end{array}
$$}

\textcolor{blue}{Since all matrices $\Phi_i(A_{\nu(i)},B_{\nu(i)})$ in \eqref{eq:WiPsi0} are positive definite, the last term of the previous expression is negative definite (by application of the Schur Complement).  This implies that the following inequality holds 
$$
\begin{array}{lcl}
V(x^+)&\!\leq &\!V(x)- \mu(V(x)\!-\!1)-\delta(\lambda^2-w^\top w) \leq \!V(x)- \mu(V(x)\!-\!1)=(1-\mu)V(x) -\mu
\end{array}
$$
where the last inequality holds because $w^\top w\leq \lambda^2$ holds by assumption. The proof is concluded by recalling that $V(x)\leq1$ and $\mu\in (0,1)$, which ensure $V(x^+)\leq \!1\!-\!\mu \!+\! \mu\leq1$.
}

\end{proof}

Before going further on the optimization procedure or the extension to the data-driven case, several comments should be properly stated to highlight the contributions of Theorem~\ref{th:theo_robust0} with respect to the existing literature on switched affine systems.

First, it should be noted that the LMI conditions presented in Theorem~\ref{th:theo_robust0} are affine and consequently convex with respect to the system matrices $(A_j,B_j)_{j \in\mathbb K}$. This structure allows for a direct extension to the case of uncertain mode matrices. The previous theorem ensures that ellipsoids $\mathcal {E}(W_i^{-1},\zeta_i)$ verify inclusions \eqref{Def:robustLimitCycle} that characterize robust limit cycles (as shown in Fig.~\ref{fig:robust_stabilization}), whenever matrices $(A_j,B_j)_{j \in\mathbb K}$ are known and constant or not. When comparing the first and recent paper on the convergence of switched affine systems to limit cycles \cite{egidio2020global}, the control law proposed therein is a time-varying state feedback that requires exact knowledge of the system matrices. Therefore, this method is not applicable to the problem under consideration and, more generally, it is impossible to extend the solution provided in \cite{egidio2020global} to the case of uncertain systems. 
\begin{figure}
		\centering
		       \hspace{-0.35cm} \includegraphics[width=0.4\textwidth]{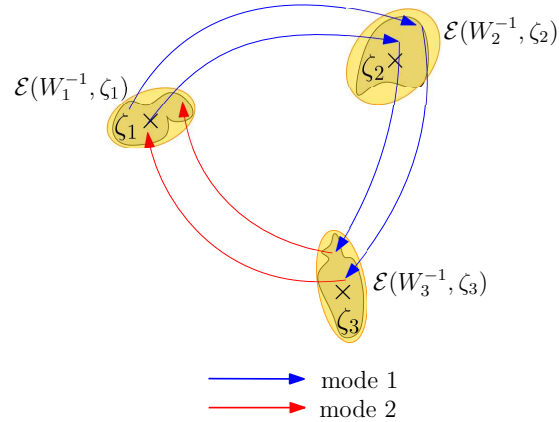}		  
	\caption{Outer estimation of the robust limit cycle thanks to attractor $\mathcal S_\nu$ expressed as the union of shifted ellipsoids characterized by $\{(W_i,\zeta_i)\}_{i\in \mathbb D_\nu}$, the solution to the LMI problem~\eqref{eq:th_LMI_rob}.}
		\label{fig:robust_stabilization}
	\end{figure}
Furthermore, the added value of the present contribution with respect to \cite{serieye2022attractor} refers to the inclusion of an external disturbance $w$, which is assumed to be bounded but not vanishing as time increases. This external noise affects the size of the ellipsoids composing the attractor in the sense that the larger the amplitude of the noise, the larger the size of the attractor.

In addition, the authors of \cite{egidio2020global} have also considered the situation of switched affine systems subjected to an external disturbance. However, the disturbance function is assumed to be $L_2$, so that an $L_2$ performance analysis was performed. This is a stronger assumption compared to the paper in hand. Indeed, having bounded but not vanishing disturbances prevents stabilization to a limit cycle (i.e., the union of singletons) but allows it in a robust limit cycle (i.e. the union of ellipsoids). The case of $L_2$ performance analysis can be easily processed using the usual LMI techniques and is therefore not provided in this paper. 

It should be noted that if the matrices of all modes are constant and known, the decision variables $\zeta_i$'s can be replaced by $\rho_i$' s, the solution to \eqref{def:cycleq}. Then, the off-diagonal block $(2,4)$ and $(4,2)$ of $\Phi(A_{\nu(i)},B_{\nu(i)})$ equal to zero, by definition, and several simplifications of the conditions can be performed. Indeed, the application of the Schur complement twice, condition $\Phi(A_{\nu(i)},B_{\nu(i)})\succ 0$ is equivalent to 
$$
A_{\nu(i)}W_i A_{\nu(i)}^\top - (1\!-\!\mu) W_{\lfloor i+1\rfloor_\nu} +\tilde \delta_i I \prec 0,\ \tilde \delta_i \!>\! \frac{(1 \!-\! \mu)\lambda^2}{\mu},\ \forall i\in\mathbb D_\nu,
$$
where $\tilde \delta_i$ is a new decision variable that stands for $(1-\mu) \delta_i^{-1}$ and where we recall that $\mu$ is a tuning parameter that must be fixed a priori.

Finally, a deeper discussion of the stabilization to robust limit cycles conditions has been provided in \cite{serieye2022attractor}, going beyond the results provided in \cite{egidio2020global}, such as, for instance, the necessary and sufficient conditions for the existence of a limit cycle and the comparison between the time-varying state feedback controller and the pure state feedback law \eqref{eq:control_law_robust0}. In addition, as in \cite{egidio2020global}, a general optimization problem has been presented thanks to the definition of a generic cost function, which aims to characterize the distance of the attractors to a desired reference, the amplitude of the limit cycles, etc. This problem allows for selecting the optimal cycle among a set of possible cycles that minimizes this cost function. This discussion is not presented in this paper to avoid repetition with \cite{serieye2022attractor}. 

\textcolor{blue}{
\begin{remark}
The complexity of the LMI condition of Theorem \ref{th:theo_robust0} relies on the number of decisions variables ($(n(n+1)/2+n+1) N_\nu$) and on the dimension of the condition ($(3n+1) N_\nu$). Noting that  the length of the limit cycle under consideration, $N_\nu$, is limited to $10$ in practice, the complexity of the LMI of Theorem \ref{th:theo_robust0} is reasonable.
\end{remark}
}

\section{Data-driven control for switched affine systems}

\subsection{Measurement noise and modelling of uncertainties}

Unlike in \cite{serieye2022attractor}, the matrices that define the modes of \eqref{eq:model_x} are not assumed to be known here. Only some experimental data are available and will be used to design the control law. Following the method presented in \cite{berberich2020robust,van2020data}, we define the following matrices: 
\begin{equation}\label{def:X_j}
\begin{array}{lcllcl}
X^+_j&:=&\begin{bmatrix} x^+_{j,1}& x^+_{j,2}&\dots& x^+_{j,p_j}
\end{bmatrix}, &\quad 
X_j&:=&\begin{bmatrix} x_{j,1}& x_{j,2}&\dots& x_{j,p_j}
\end{bmatrix}.
\end{array}
\end{equation}

Matrices $X_j^+\in\mathbb R^{n\times p_j}$ and $X_j\in\mathbb R^{n\times p_j}$ collect all the data from the experiments obtained for several initial conditions $x_{j,\ell}$. Subscripts `$j$' and `$\ell$' in $x_{j,\ell}$ refer to the mode under consideration and the index of the experiment, respectively. Notation $\ell$ does not necessarily refer to time here. In fact, experiments can be built using arbitrary vectors $x_{j,\ell}$, or selecting them such that $x^+_{j,\ell}=x_{j,\ell+1}$ for all $\ell=1,\dots,p_j-1$. \newline An experiment $(x^+_{j,\ell},x_{j,\ell})$ verifies the following equation
$$
x^+_{j,\ell}=A_jx_{j,\ell}+B_j +w_{j,\ell},\quad \forall \ell \in [1,p_j]_\mathbb N,\quad \forall   j\in \mathbb K.
$$ 

Differently from the usual linear case \cite{berberich2020robust,berberich2020combining,vanWaarde2020noisy,van2020data}, one has to differentiate the experiments performed for each mode $j$ in $\mathbb K$.  In the previous equation, the experiments have been corrupted by the measurement noise represented by the additional vector $w_{j,\ell}$, for all $j,\ell$ in $\mathbb K\times[1,p_j]_\mathbb N$ that is also gathered in matrix $\omega_j$ defined as follows
$$\omega_j=\begin{bmatrix} w_{j,1}& w_{j,2}&\dots& w_{j,p_j}
\end{bmatrix}\in\mathbb R^{n\times p_j}.$$ 

Summing up all these ingredients, the experimental data verify the following equation.
\begin{equation}\label{eq:info_data}
X^+_j=A_j X_j+B_j \mathbf{1}^\top_{p_j} +\omega_j, \quad \forall j\in \mathbb K.
\end{equation}

Again, several definitions inspired by \cite{van2020data} need to be stated to properly formulate the control problem.

\begin{definition}\cite{van2020data}\label{def:Y_j}
Consider the data matrix $Y_j$ associated to each mode defined as follows
\begin{equation}\label{eq:def:Y_j}
Y_j = \begin{bmatrix}
I_n&X_j^+\\
0&-X_j\\
0 & -\mathbf{1}_{p_j}^\top
\end{bmatrix}, \quad \forall j \in \mathbb K.
\end{equation}
\end{definition}

Matrices $Y_j$ are called \textit{data matrices} and will be of fundamental interest throughout the paper. Its relevance is demonstrated by noting that
 \begin{equation}\label{Eq:UncMat-Y}
  \left[\begin{matrix}
I_n& A_j &B_j
\end{matrix}\right]Y_j=  \left[\begin{matrix}
I_n& \omega_j\end{matrix}\right].
 \end{equation}
 
Indeed it reflects the relationship between the uncertain matrices of the system, i.e. $A_j,B_j$, and the measurement noise associated with these experiments, i.e. $\omega_j$. Equation \eqref{Eq:UncMat-Y} also provides some information on the relevance of the information. Essentially, matrices $Y_j$ need to include a sufficiently large number of experiments in order to be used in the control design phase. This intuition is formalized in the next definition taken from \cite{vanWaarde2020noisy,van2020data} on the informative nature of matrices $Y_j$ \cite{berberich2020robust,berberich2020data}. 
 \begin{definition}\label{def:informativity}
A matrix $Y_j$ given in \eqref{eq:def:Y_j} is said to be informative if $Y_jY_j^\top$ is nonsingular (or equivalently positive definite).  
 \end{definition}
\begin{definition}
Data $\{X_j\}_{j\in\mathbb K}$ are called persistently
exciting for the switched affine system~\eqref{eq:model_x} if matrices $\begin{bsmallmatrix}
X_j\\ \mathbf{1}_{p_j}^\top\end{bsmallmatrix}$ have full row rank.
\end{definition}

\begin{lemma}
Matrices $\{Y_j\}_{j\in\mathbb K}$ are informative only if and only if $\{X_j\}_{j\in\mathbb K}$ are persistently exciting for the switched affine system~\eqref{eq:model_x}.
\end{lemma}
\begin{proof}
Assume that matrices $\{Y_j\}_{j\in\mathbb K}$ are informative, this means that $Y_jY_j^\top$ are non-singular. More particularly, the lowest $(n+1)\times (n+1)$ diagonal block are given by $\begin{bsmallmatrix}
X_j\\ \mathbf{1}_{p_j}^\top\end{bsmallmatrix}\begin{bsmallmatrix}
X_j\\ \mathbf{1}_{p_j}^\top\end{bsmallmatrix}^\top$, which are nonsingular as well. Therefore, matrices $\{X_j\}_{j\in\mathbb K}$ are persistently excited. The reverse implication is direct by noting that if $\{X_j\}_{j\in\mathbb K}$ are assumed to be persistently excited, then matrices $Y_j$ are full row rank because of the structure of $Y_j$ in \eqref{eq:def:Y_j}. 
\end{proof}

\begin{assumption}\label{Assumption0}
Matrices $\{Y_j\}_{j\in\mathbb K}$ are informative. 
\end{assumption}
In the sequel, we will use the same lines as those that have been introduced in \cite{vanWaarde2020noisy,van2020data}, consisting of imposing the following assumption on the data noise. 

\begin{assumption}\label{Assumption1}
For all $j$ in $\mathbb K$, there exists a symmetric matrix $\Psi_j=\begin{bsmallmatrix} \Psi_{j1}&\Psi_{j2}\\
 \ast &\Psi_{j3}\\
\end{bsmallmatrix}
$ in $\mathbb S^{n+p_j}$ such that $\Psi_{j3}\prec 0$ and 
\begin{equation}\label{eq:ineqW}
\begin{bmatrix} I_n\\ \omega_j^\top\end{bmatrix}^\top\Psi_j
\begin{bmatrix} I_n\\ \omega_j^\top\end{bmatrix} \succeq 0, \quad 
\end{equation}
for the data noise $\omega_j$ satisfying \eqref{eq:info_data}.
\end{assumption}

Following the setup presented in Section~\ref{Sec:Prel},  \eqref{eq:ineqW}   in Assumption~\ref{Assumption1}  is rewritten following the form given in \eqref{def_Sigma}, i.e., as an inequality expressed in terms of matrices $A_j,B_j$ instead of $\omega_j$. Therefore, we now define the set $\Sigma(\Psi_j,Y_j)$ as follows:
 \begin{equation}\label{def_Sigmaj}
 \Sigma(\Psi_j,Y_j):=\left\{
 [A_{\!j} \ B_{\!j}]\!\in\!\mathbb R^{n\times (n\!+\!1)} \mbox{ s.t. } \left[\begin{smallmatrix}
I\\
A_j^\top\\
B_j^\top 
\end{smallmatrix}\right]^{\!\top}Y_j \Psi_jY_j^\top \left[\begin{smallmatrix}
I\\
A_j^\top\\
B_j^\top 
\end{smallmatrix}\right]\succeq 0
 \right\}.
 \end{equation}

As in Section~\ref{Sec:Prel}, this set represents all the possible pairs of matrices $A_j$ and $B_j$ in $\mathbb R^{n\times n}$ and $\mathbb R^{n\times 1}$ that verify \eqref{eq:info_data} and where the data noises $\omega_j$ verify Assumption~\ref{Assumption1}.

Summing up the previous ingredients, the problem can be formulated as follows.
\begin{problem}
Find decision variables $\{(W_i,\zeta_i)\}_{i\in \mathbb D_\nu}$ in $\mathbb{S}^{n}\times \mathbb R^{n}$ such that, for all $i$ in $\mathbb D_\nu$
\begin{equation}\label{Pb:DataDrivenSAS}
\Phi_i(A_{\nu(i)}, B_{\nu(i)})\succ0, \quad \forall (A_{\nu(i)}, B_{\nu(i)})\in \Sigma(\Psi_{\nu(i)}, Y_{\nu(i)}).
\end{equation}
\end{problem}
This problem fits exactly in the framework presented in Section~\ref{Sec:Prel}, and a solution is provided in the theorem below.
More precisely, if such a problem admits a solution, then it is possible to build a stabilizing control law such that the closed-loop trajectories of the switched affine system \eqref{eq:model_x} globally and asymptotically  converge to an outer estimation of the robust limit cycle $\mathcal S_\nu$, composed of the union of shifted ellipsoids, characterized by $\{(W_i,\zeta_i)\}_{i\in \mathbb D_\nu}$, solution to the LMI problem~\eqref{Pb:DataDrivenSAS}. 

It is also worth noting that the previous problem refers to a data-driven control design, since matrices $A_j,B_j$ are not involved in the previous problem, which only requires data experiments.

\subsection{Main result}

This section provides a new contribution to the data-based design of the stabilizing switching control law for switched affine systems. In this situation, the lack of knowledge of the system matrices $A_j,B_j$ prevents the use of the method provided in \cite{egidio2020global}, since, again, the control law developed therein depends explicitly on $A_j,B_j$. Furthermore, since the data are corrupted by external disturbance $\omega$ (if $\lambda\neq0$), it is not possible to find an exact expression of $A_j,B_j$ based on the data, as the data only provide a set of uncertain allowable matrices $A_j,B_j$. Therefore, the solution provided in \cite{egidio2020global} is not applicable to solve this problem. However, we will demonstrate below how Theorem~\ref{th:theo_robust0} can be easily adapted to data-driven design thanks to the use of  Lemma~\ref{lem0}.

\begin{theorem}\label{th:theo_data-robust}
	For a selected cycle $\nu$ in $\calc$ and for a given parameter $\mu \in \left(0,1\right)$, under Assumptions \ref{Assumption0} and \ref{Assumption1} for given matrices $\{\Psi_j\}_{j\in\mathbb K}$, consider $\{(W_i,\zeta_i,\eta_i,\delta_i)\}_{i\in \mathbb D_\nu}$ in $\mathbb{S}^{n}\times \mathbb R^{n}\times (\mathbb R_{>0})^2$, the solution to the following problem.
\begin{equation}\label{th:LMI_data}
	\begin{array}{rcl}
\bar \Phi_i(Y_{\nu(i)})  \succ0, \ \ W_i \succ0, \ \ \delta_i>0,\ \ \eta_i>0,\ \  \forall i\in \mathbb D_\nu,
	\end{array}
	\end{equation}
with
\begin{equation}\label{eq:WiPsi2}
	\bar \Phi_i(Y_{\nu(i)}) \!= \!
	\left[
\begin{array}{ccc|ccc}
 \!(1 \!-\! \mu)W_i & 0 	&0	  &  \quad 0  &  \quad W_i&0 \\ 
\ast 			& \! \! \! \!\mu\!-\!\delta_i \lambda^2 &0	& \quad -\zeta_{\lfloor i+1\rfloor_\nu}^\top & \quad   \zeta_i^\top &1\\
\ast&\ast & \delta_i I&\quad I& \quad 0&0\\
\hline 
\ast&\ast&\ast&  \multicolumn{3}{c}{\multirow{3}{*}{$
\begin{bsmallmatrix} W_{\lfloor i+1\rfloor_\nu}& 0\\ 0 & 0 \end{bsmallmatrix} -\eta_i Y_{\nu(i)}\Psi_{\nu(i)}Y_{\nu(i)}^\top
$}} \\
\ast&\ast &\ast &   \\
\ast&\ast &\ast &   \\
\end{array} 
\right]
\end{equation}
which depends only on the decision variables and on the data collected in the augmented vector data $Y_j$. 
	Then 
		attractor \begin{equation}\label{def:S}
			\mathcal{S}_\nu:=\bigcup_{i\in\mathbb D_\nu}\mathcal E( W_i^{-1}, \zeta_i) 
		\end{equation}
		where $\mathcal E( W_i^{-1}, \zeta_i)=\left\lbrace x \in \mathbb{R}^n, i\in\mathbb D_\nu \big| (x-\zeta_i)^\top W^{-1}_i(x-\zeta_i) \leq 1 \right\rbrace,$ is robustly globally exponentially stable for system~\eqref{eq:model_x} with the disturbance signal \eqref{eq:cond_w} and with
		the switching control law
		\begin{equation}\label{eq:control_law_robust}
			u(x) = \left\{ \nu\left( \theta \right),~\theta\in\underset{i\in\mathbb D_\nu}{\argmin}\left(x-\zeta_i\right)^\top W_i^{-1}\left(x-\zeta_i\right) \right\}\subset \mathbb K.
		\end{equation}

\end{theorem}

\begin{proof} The idea of the proof is to prove that conditions \eqref{eq:WiPsi2} and \eqref{Pb:DataDrivenSAS} are equivalent, thanks to the use of  Lemma~\ref{lem0} in this specific context. To do so, we need to understand how inequality \eqref{Pb:DataDrivenSAS} verifies the requirements and assumptions of Lemma~\ref{lem0}. In particular, the proof will be divided into the three following steps:
\begin{enumerate}[label=(\roman*)]	
\item Rewrite inequality $\Phi_i(A_{\nu(i)}, B_{\nu(i)})\succ0$ as in \eqref{lem0_1}.
\item \textcolor{blue}{Verify that conditions $\Psi_{j3}
:=\left[\begin{smallmatrix}
0_{n,n+1}\\
I_{n+1}\\
\end{smallmatrix}\right]^\top  Y_j\Psi_jY_j ^\top \left[\begin{smallmatrix}
0_{n,n+1}\\
I_{n+1}\\
\end{smallmatrix}\right]\preceq 0 $, for all $j \in \mathbb K$ hold.}
\item Apply Lemma~\ref{lem0}.
\end{enumerate}

\textit{Proof of (i):} To do so, let us note that $\Phi_i (A_{\nu(i)}, B_{\nu(i)})$ from Theorem \ref{th:theo_robust0} can be rewritten as follows
$$
\begin{array}{lclllcl}
		\Phi_i (A_{\nu(i)}, B_{\nu(i)}) 
		&=&\left[\begin{array}{c|c}
		    	\begin{bmatrix} (1 \!-\! \mu)W_i &\!0  &0\\
			\ast &\! \mu\!-\!\delta_i \lambda^2 &0\\
			\ast&\ast &\delta_i I \end{bmatrix} & \begin{bmatrix}
			W_i A_{\nu(i)}^\top \\ 
			\! \zeta_i^\top A_{\nu(i)}^\top +B_{\nu(i)}^\top-\zeta_{\lfloor i+1\rfloor_\nu}^\top\\
			I\end{bmatrix}\\ \hline 
		    	\ast  &\! W_{\lfloor i+1\rfloor_\nu} 
		    \end{array}\right]\\ \\
		    &=&\left[\begin{array}{c|c}
		    \begin{bmatrix} (1 \!-\! \mu)W_i &\!0  &0\\
			\ast &\! \mu\!-\!\delta_i \lambda^2 &0\\
			\ast&\ast &\delta_i I \end{bmatrix} & 
			 \begin{bmatrix} 
			0 \\
			-\zeta_{\lfloor i+1\rfloor_\nu}^\top\\I \end{bmatrix}
+\begin{bmatrix} 
			 W_i &0 \\
			 \zeta_i^\top &1\\
			 0&0\end{bmatrix} \left[\begin{matrix}
A_j^\top\\
B_j^\top 
\end{matrix}\right]\\ \hline
		    	\ast  &\! W_{\lfloor i+1\rfloor_\nu} 
		    \end{array}\right].
\end{array}
$$

Therefore, we can identify that 
$$\begin{array}{lclllclllcllclllclllcl}
\mathcal M_1&=&\begin{bmatrix} (1 \!-\! \mu)W_i &\!0 &0 \\
			\ast &\! \mu\!-\!\delta_i \lambda^2 &0\\
			\ast&\ast &\delta_i I \end{bmatrix},& 
\mathcal M_2&=&W_{\lfloor i+1\rfloor_\nu},\
\mathcal N_1 &=&\begin{bmatrix} 
			0 \\
			-\zeta_{\lfloor i+1\rfloor_\nu}^\top\\I \end{bmatrix},
\mathcal N_2 =\begin{bmatrix} 
			 W_i &0 \\
			 \zeta_i^\top &1\\
			 0&0\end{bmatrix},&\mathcal A^\top &=& \left[\begin{matrix}
A_j^\top\\
B_j^\top 
\end{matrix}\right].
\end{array}
$$
Lemma~\ref{lem0} ensures that having $\Phi_i (A_{\nu(i)}, B_{\nu(i)})\succ0$ is equivalent to $\bar \Phi_i(Y_{\nu(i)})\succ0$, for all $i\in\mathbb D_\nu$, which concludes this part of the proof.
\newline
\textit{Proof of (ii):} Comparing the definition of $\Sigma(\Psi)$ in \eqref{def_Sigma} and $\Sigma(\Psi_j,Y_j)$ in \eqref{def_Sigmaj}, it yields $\Psi=Y_j\Psi_j Y_j^\top=\begin{bsmallmatrix}
\Psi_{j1}&\Psi_{j2}\\\Psi_{j2}^\top&\Psi_{j3}
\end{bsmallmatrix}$ for any $j$ in $\mathbb K$. Therefore, to apply lemma~\ref{lem0}, we need to verify that 
 $\Psi_{j3}\prec 0$, for all $j \in \mathbb K$. To do so, let us recall that, in Assumption~\ref{Assumption1}, conditions $\Psi_{j3}\leq0$ ensure that there exist, for any $j$ in $\mathbb K$, a symmetric positive definite matrix $0\prec R_j\in\mathbb S^{n}$ and a scalar $\epsilon_j$ such that matrix 
$\Psi_j- \begin{bsmallmatrix}R_j&0\\0&0 \end{bsmallmatrix}\prec -\epsilon_j I.$
Pre- and post-multiplying this inequality by $Y_j$ and its transpose, respectively,
we have 
$$
Y_j\left( \Psi_j- \begin{bsmallmatrix}R_j&0\\0&0 \end{bsmallmatrix}\right)Y_j^\top \prec -\epsilon_j Y_jY_j^\top.
$$
Since the data matrix $Y_j$ is assumed to be informative, i.e. matrices $Y_jY_j^\top$ are non-singular, there exists $\bar\epsilon_j>0$ such that $Y_j \left(\Psi_j-\begin{bsmallmatrix}R_j&0\\0&0 \end{bsmallmatrix}\right)Y_j^\top \prec -\bar\epsilon_j I$. Furthermore, recalling the structure of the data matrix $Y_j$ that makes  $\begin{bsmallmatrix}
0_{n,n+1}\\
I_{n+1}\\
\end{bsmallmatrix}^\top Y_j\begin{bsmallmatrix}R_j&0\\0&0 \end{bsmallmatrix}Y_j^\top\begin{bsmallmatrix}
0_{n,n+1}\\
I_{n+1}\\
\end{bsmallmatrix}^\top=0$, the following is finally gotten 
$$
\left[\begin{matrix}
0_{n,n+1}\\
I_{n+1}\\
\end{matrix}\right]^\top  Y_j\Psi_jY_j \left[\begin{matrix}
0_{n,n+1}\\
I_{n+1}\\
\end{matrix}\right]\prec -\bar\epsilon_j \left[\begin{matrix}
0\\
I_{n+1}\\
\end{matrix}\right]^\top \left[\begin{matrix}
0\\
I_{n+1}\\
\end{matrix}\right]
= -\bar\epsilon_j 
I_{n+1} \textcolor{blue}{\prec 0},$$
which was to be proven.  
\end{proof}

	The previous theorem allows to design a control law that stabilizes system \eqref{eq:model_x}, to the attractor defined by set $\mathcal S_\nu$ given in \eqref{def:S}, which is a union of shifted ellipsoids $\mathcal S_\nu=\cup_{i\in\mathbb D_\nu}\mathcal E( W_i^{-1}, \zeta_i)$. 

\textcolor{blue}{
\begin{remark}
The complexity of the LMI condition of Theorem \ref{th:theo_data-robust} relies on the number of decisions variables ($(n(n+1)/2+n+2) N_\nu$) and on the dimension of the condition ($(4n+2) N_\nu$). Both numbers are lightly higher than the ones of Theorem \ref{th:theo_robust0}. Again the length of the limit cycle under consideration, $N_\nu$ is limited to $10$ in practice, then the complexity of the LMI of Theorem \ref{th:theo_data-robust} is  reasonable.
\end{remark}
}

\begin{figure}[t]
		\centering
		  \includegraphics[width=0.45\textwidth]{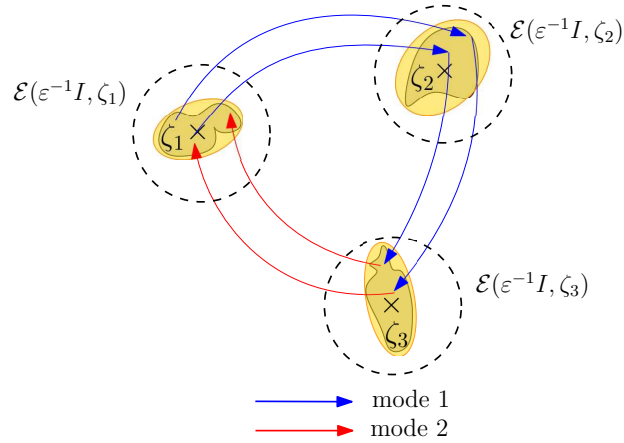}		  

	\caption{Principles of the optimization of the attractor $\mathcal S_\nu$ thank to the minimization of the a decision variable $\varepsilon$.}	\label{fig:optimization}
	\end{figure}

\section{Optimization}
Theorems~\ref{th:theo_robust0} and~\ref{th:theo_data-robust} present a model-based and a data-based condition for the stabilization and the existence of an attractor for the closed-loop system. Both conditions are expressed in terms of LMIs, whose solutions impact both the control laws and the size of the attractors notably through the matrices $W_i$, $i\in\mathbb D_\nu$. It is then natural to ask for the optimization of these solutions in order to minimize the size of the attractors, which can be achieved by ``minimizing'' matrices $W_i$.

\subsection{Optimization of the attractor}

 There are many ways to perform such minimization. One of them could be to introduce a cost function consisting in the sum of the trace of these matrices. Here we rather consider the additional constraint $W_i\prec \varepsilon I_n$, for some positive scalar $\varepsilon$. This additional inequality can be interpreted as constraining the ellipsoids $\mathcal E(W_i,\zeta_i)$ to be all included in the ball of radius $\sqrt{\varepsilon}$, centered at the same location (see Fig.~\ref{fig:optimization}). In other words, we impose \begin{equation}\label{eq:inclusion_opt0}
\mathcal E(W_i^{-1},\zeta_i)\subset \mathcal E(\varepsilon^{-1} I_n,\zeta_i),\quad \forall i\in \mathbb D_\nu.
 \end{equation} 
An additional feature can be added to reduce the ``size'' of this attractor by introducing an optimization problem, as described below.
\begin{opt}
For a given $\mu$, the optimal solution to the condition of Theorem~\ref{th:theo_robust0} that verifies inclusion \eqref{eq:inclusion_opt0} with the smallest $\varepsilon$ is obtained by solving the following optimization problem
\begin{equation}\label{opt:model}
	\begin{array}{rcl}
\varepsilon^\ast &=&\displaystyle \min_{\displaystyle\{ (W_i,\zeta_i,\delta_i)\}_{i\in \mathbb D_\nu}} \varepsilon\\
	&\mbox{s.t.}& \Phi_i(A_{\nu(i)},B_{\nu(i)})  \succ0, \ \varepsilon I \succ W_i \succ0, \ \delta_i>0, \quad \forall i\in \mathbb D_\nu,
	\end{array}
	\end{equation}
	where $\Phi_i(A_{\nu(i)},B_{\nu(i)})$ is given in \eqref{eq:WiPsi0}.
\end{opt}

Similarly, an optimization procedure can be added to the conditions of Theorem~\ref{th:theo_data-robust} as provided below.
\begin{opt}
For a given $\mu$, the optimal solution to the condition of Theorem~\ref{th:theo_data-robust}, that verifies inclusion \eqref{eq:inclusion_opt0} with the smallest $\varepsilon$ is obtained by solving the following optimization problem
\begin{equation}\label{opt:data}
	\begin{array}{rcl}
\varepsilon^\ast &=&\displaystyle \min_{\displaystyle\{ (W_i,\zeta_i,\eta_i,\delta_i)\}_{i\in \mathbb D_\nu}} \varepsilon\\
	&\mbox{s.t.}&\bar \Phi_i(Y_{\nu(i)})  \succ0, \ \varepsilon I \succ W_i \succ0, \ \delta_i>0, \ \eta_i>0,\ \forall i\in \mathbb D_\nu,
	\end{array}
	\end{equation}
	where $\bar \Phi_i(Y_{\nu(i)})$ is given in \eqref{eq:WiPsi2}.
\end{opt}

Introducing constraints $\varepsilon I \succ W_i$ and minimizing $\varepsilon$ refer to a usual optimization to reduce the size of the components of the attractors as for Theorem~\ref{th:theo_robust0}.

Apart from the usual causes associated with Lyapunov and LMI techniques, two reasons may explain this conservatism in the numerical result. The first is related to the fact that the optimization problem presented in Theorem~\ref{th:theo_robust0} requires the use of two successive S-procedures. It is well known that this introduces conservatism compared to the situation where only one is performed. Second, the optimal value $\varepsilon^\ast$ strongly depends on parameter $\mu$, which has been fixed a priori.
\textcolor{blue}{It would be possible to refine the values obtained for $\varepsilon$ by tuning this parameter, using a ``line-search'' algorithm on $\mu$.}

\subsection{Optimization selection of the cycle}

\textcolor{blue}{Following the idea first presented in \cite{egidio2020global}, it is possible to introduce a cost function that aims at evaluating the performance of a cycle. For instance, this cost function should reflect for each cycle the chattering effects, which means the amplitude of the trajectories within the (robust) limit cycle, or the distance of the (robust) limit cycle to a desired functioning point $x_d$ in $\mathbb R^n$, selected by the designer. However, the solution proposed in \cite{egidio2020global} relies on the exact knowledge of the limit cycle, i.e., $\{\rho_i\}_{i\in\mathbb D_\nu}$ generated by a given cycle $\nu$ in the situation of system \eqref{eq:model_x} without disturbances. }

\textcolor{blue}{When considering robust limit cycles and its outer estimation using the attractor $\mathcal S_\nu$, such a procedure cannot be employed directly, since the estimation of the robust limit cycle is provided by the solution of the LMI optimization problems \eqref{opt:model} and \eqref{opt:data}. Instead, the solution provided in \cite{serieye2022attractor} allows formulating such a cost function based on the decision variables of the LMI optimization problem. For the sake of readability and to avoid an overlap with this contribution, such extension is not presented here but the reader can refer to \cite[Section 6]{serieye2022attractor}.}

\section{Numerical applications}
\label{sec:sim}

\begin{figure*}
     \centering
     \begin{subfigure}[b]{\textwidth}
         \centering
  \includegraphics[width=0.85\textwidth]{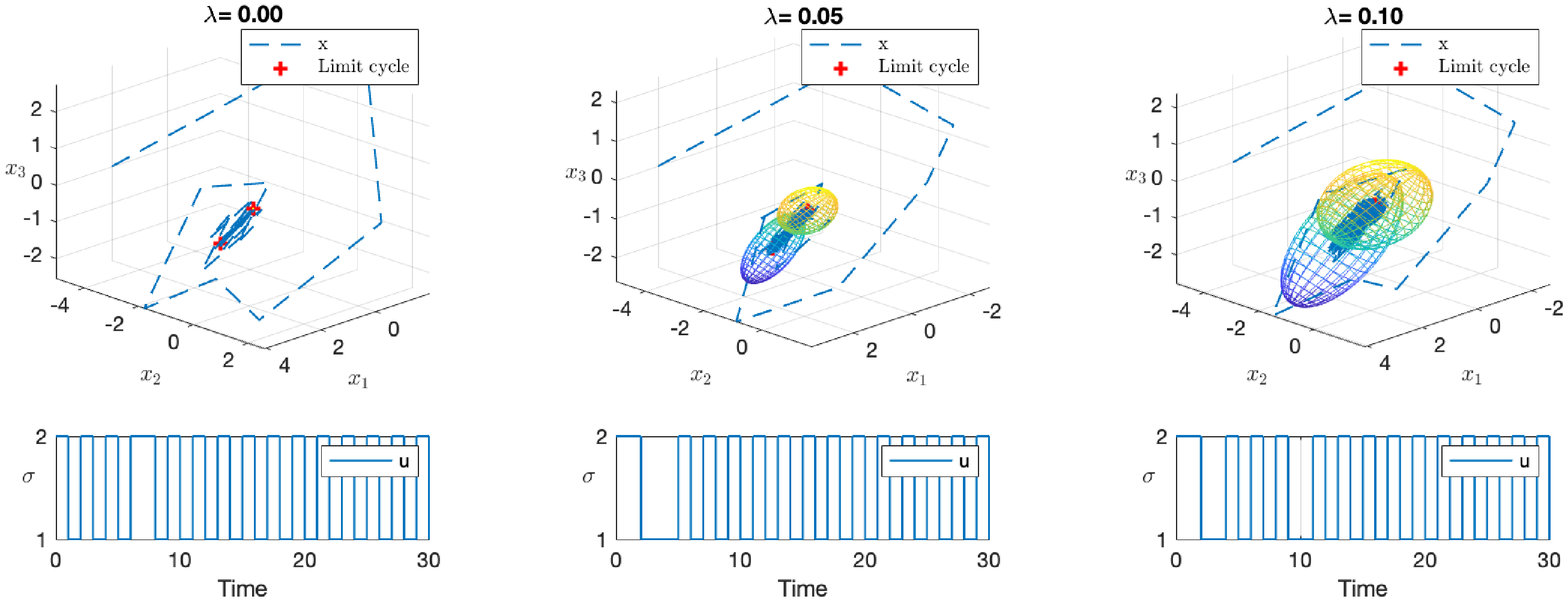}
         \vspace{-0.2cm}
         \caption{Solutions obtained for cycle $\nu_1=\{1,2\}$.}
         \label{fig:model-basednu1}
     \end{subfigure}
     \hfill
     \begin{subfigure}[b]{\textwidth}
         \centering
         \includegraphics[width=0.85\textwidth]{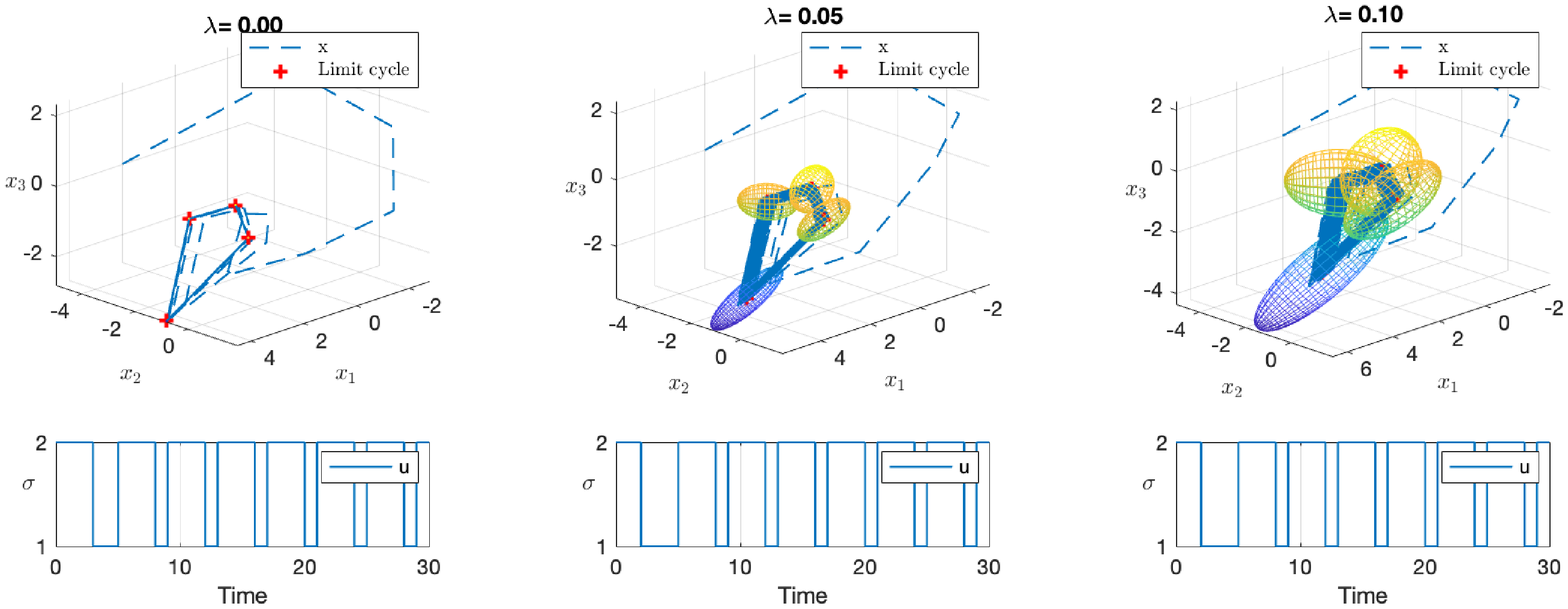}
         \vspace{-0.2cm}
         \caption{Solutions obtained for cycle $\nu_2=\{1,2,2,2\}$.}
         \label{fig:model-basednu2}
     \end{subfigure}
     \hfill
     \begin{subfigure}[b]{\textwidth}
         \centering
         \includegraphics[width=0.85\textwidth]{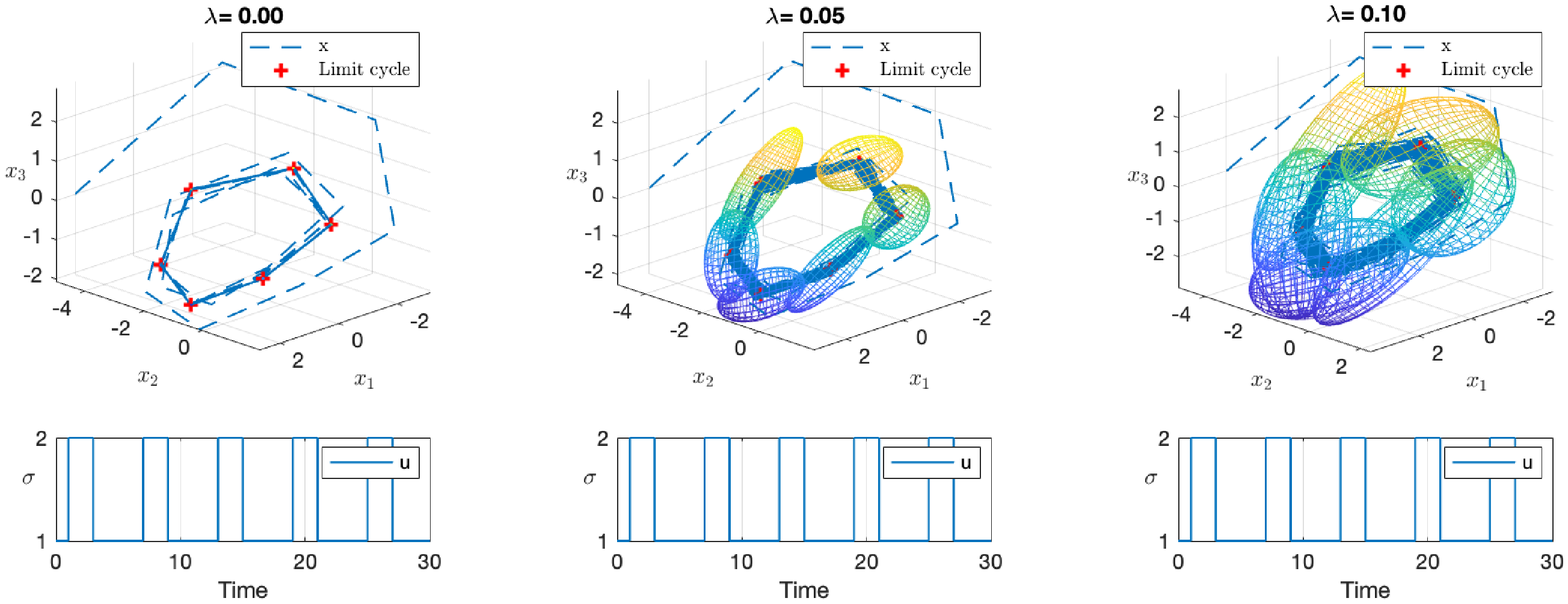}
         \vspace{-0.2cm}
         \caption{Solutions obtained for cycle $\nu_3=\{1,1,1,1,2,2\}$.}
         \label{fig:model-basednu3}
     \end{subfigure}
        \caption{State and input trajectories of the closed-loop system for three different cycles obtained by solving Theorem~\ref{th:theo_robust0} for $\lambda=0$, $0.05$ and $0.1$ with $\mu=0.1$ with the initial condition $x_0=[2 \ -5 \ 0]^\top$, for all the cycles and bounds $\lambda$, considered here.}
        \label{fig:model-based}
\end{figure*}

\subsection{System data}
To illustrate the two results presented in this paper, here we will consider an example of switched affine systems \eqref{eq:model_x} borrowed from \cite{deaecto2016stability},  where the matrices $A_i$ and $B_i$ are defined as follows. 
\begin{equation}\label{eq:discretize}
		A_i = e^{F_i T}, \quad B_i = \int_0^T e^{F_i \tau}\text{d}\tau g_i, \quad \forall i \in \left\lbrace 1,2 \right\rbrace,
\end{equation}
where $T=1$ is the sampling period of the associated continuous-time system $\dot x=F_jx+g_j$ defined with matrices $F_j$ and $g_j$, with $j\in \mathbb K=\{1,2\}$ given by 
\begin{equation*}
	\begin{array}{llll}
		F_1\! =\! 
		\left[\begin{smallmatrix}
			\ 0&\ 1&\ 0\\
			\ 0&\ 0&\ 1\\
			\! -1&\! -1&\! -1
		\end{smallmatrix}\right], & F_2\! =\! 
		\left[\begin{smallmatrix}
			0&\ 1&\ 0\\
			0&\ 0&\ 1\\
			0&\! -1&\! -1
		\end{smallmatrix}\right], & g_1\! =\! 
		\left[\begin{smallmatrix}
			1\\
			0\\
			0
		\end{smallmatrix}\right], & g_2 \!=\! 
		\left[\begin{smallmatrix}
			0\\
			1\\
			0
		\end{smallmatrix}\right]\!.
	\end{array}
\end{equation*}

\subsection{Model-based design: solution and simulations}
First, let us consider a model-based solution. To do so, we will consider three cycles $\nu_1=\{1,2\}$, $\nu_2=\{1, 2, 2, 2\}$, and $\nu_3=\{1,1,1,1,2,2\}$ to illustrate the potential of the result.

Solving the conditions of Theorem~\ref{th:theo_robust0}, for these three cycles and for three different values of $\lambda$, the simulations of the closed-loop systems are depicted in Figure~\ref{fig:model-based}. First, the figure shows that the attractor $\mathcal S_\nu$ is indeed composed of the union of ellipsoids. The number of ellipsoids in the attractor equals the length of the cycle. The state trajectories converge, in all cases, to the attractor. In addition, it can also be seen that the control input $u$ converges to a periodic trajectory, which is a shifted version of the desired cycle. This characteristic has been demonstrated in \cite{serieye2022attractor}, but not in this paper, as the proof is very similar. It is worth mentioning that this periodic behavior of the control input is only guaranteed when the intersection of any pair of ellipsoids composing the attractor is empty. This figure also illustrates that the size of the attractor increases with the amplitude of the disturbance $\lambda$. Indeed, the effect of the disturbance on the state trajectory is visible in this figure.

For the first and second cycles $\nu= \{1,2\}$ and $\nu= \{1,2,2,2\}$, the values of $\varepsilon^\ast$ obtained by solving Theorem~\ref{th:theo_robust0} are given in Table~\ref{Tab:eps}. These solutions show that increasing $\lambda$ naturally leads to an increase in $\varepsilon^\ast$. 

\subsection{Data-based design: Generation of the data}

In order to generate the data, the following procedure was employed. Matrix $X_{j0}$ has been selected as the solution to the system with $\sigma=j$ and an arbitrary initial condition $x_{j,0}$ with magnitude less than $1$. In other words, we have
$$
x^+_{j,\ell}=x_{j,\ell+1}=A_jx_{j,\ell}+B_j +w_{j,\ell},\quad \forall (j,\ell)\in \mathbb K\times [1,p_{max}]_\mathbb N.
$$ 
with the corresponding matrices $X_j$ and $X_j^+$ given in \eqref{def:X_j} and where the maximum number of data considered in this example is $30$. Three cases will be used, corresponding to $p_j=10$, $20$ and $30$ for all $j$ in $\mathbb K$. Note that these three cases use the same set of data.  This means that the $10$ first data from case $p_j=20$ correspond to the data used when $p_j=10$ and the $20$ first data from case $p_j=30$ are the same data as for case $p_j=20$.  

In this setup, the measurement noise has been selected as an arbitrary vector, which corrupts the solution to the noiseless system. The noise $w_j$ is shown in Figure \ref{fig:noise},  only for mode $j=1$ but for three cases of $\lambda$. The norm of the noise vectors is bounded by $\lambda$, and it is also assumed that the overall noise matrices $\omega_j$ verify 
$\omega_{j}\omega_j^\top   \leq \kappa p_j \lambda_j^2 I_n$, for all $j\in \mathbb K$, with $\kappa=0.3$. Such requirements ensure that inequality 
\begin{equation}\label{eq:ineqWexample}
\begin{bmatrix} I_n\\ \omega_j^\top\end{bmatrix}^\top\underbrace{
\begin{bmatrix}
\kappa p_j \lambda_j^2 I_n &0\\
\ast& -I_{p_j}
\end{bmatrix}}_{\Psi_j}
\begin{bmatrix} I_n\\ \omega_j^\top\end{bmatrix} \succeq 0
\end{equation}
holds for all $j\in \mathbb K$. Note that this assumption is related to the covariance matrix when $\omega_j$ is a random variable, as mentioned in \cite{vanWaarde2020noisy}. 

\begin{figure}
		\centering
		       \includegraphics[width=0.57\textwidth]{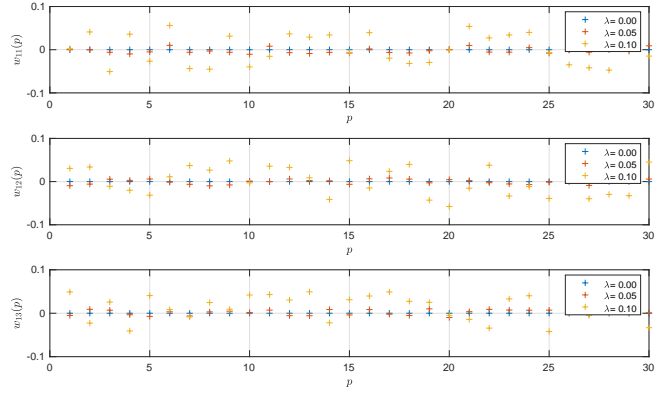}		  
	\caption{Measurement noise $\omega_1=[\omega_{11}\ \omega_{12}\ \omega_{13}]^\top$, affecting the data for mode 1. Similar noise affects the data of the other mode.}
		\label{fig:noise}
	\end{figure}
	\smallskip
	
\begin{figure*}
		\centering
		         \includegraphics[width=\textwidth]{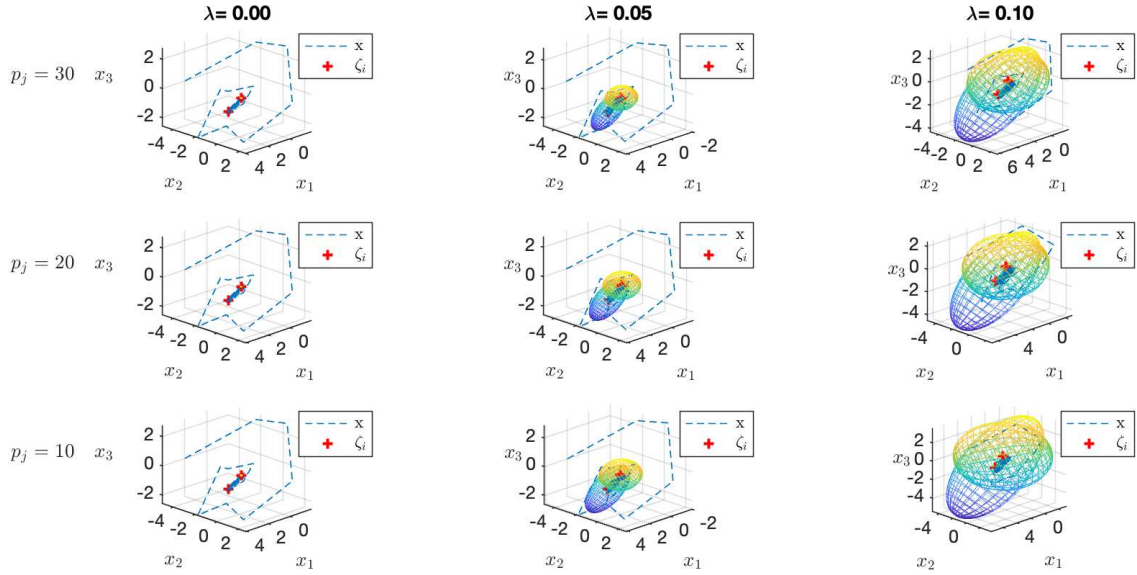}		  
		\caption{Evolution of the state variables (blue dashed line) for cycle $\nu_1=\{1,2\}$ with their associated robust limit cycles. The red crosses depict $\zeta_i$, which are the centers of the ellipsoids $\mathcal E(W_i^{-1},\zeta_i)$. From top to bottom, the plot show the attractors when the number of data available for the design are $30$, $20$ and $10$, respectively. From left to right, the uncertainty bound in the data increases $\lambda=0$, $0.05$ and $0.1$.}
		\label{fig:stab_to_limit_cycle}
	\end{figure*}
	
\subsection{Data-based design: Solution and simulations}

Applying Theorem \ref{th:theo_data-robust} and solving optimization problem~\ref{opt:model}, we have obtained the results that are resumed in Table~\ref{Tab:eps} for $\nu_1=\{1,2\}$ and the associated simulations are depicted in Figure~\ref{fig:stab_to_limit_cycle}. It can be seen from Table~\ref{Tab:eps} that augmenting $\lambda$ imposes an increase of $\varepsilon^\ast$, and that increasing the number of data leads to a decrease of the optimal value $\varepsilon^\ast$. Compared to the model-based approach, the corresponding values of $\varepsilon^\ast$ suffers from a notable increase. For the case $\lambda =0$, the large augmentation of the optimal value can be explained by the fact that Theorem~\ref{th:theo_data-robust} requires an additional S-procedure to obtain the stabilization conditions, which unavoidably introduces conservatism, as commented earlier. Nevertheless, the resulting values of $\varepsilon^\ast$ are still small and the attractors are already accurate as shown in  Figure~\ref{fig:stab_to_limit_cycle}. For the cases, where $\lambda \neq 0$, the increase of the $\varepsilon^\ast$ compared to the model-based approach can also be explained using the previous argument but also by the fact that the noise in the data also leads to additional uncertainties to the system matrices, which need to be compensated by an increase of the size of the attractor.

\begin{figure*}
\centering
\includegraphics[width=\textwidth]{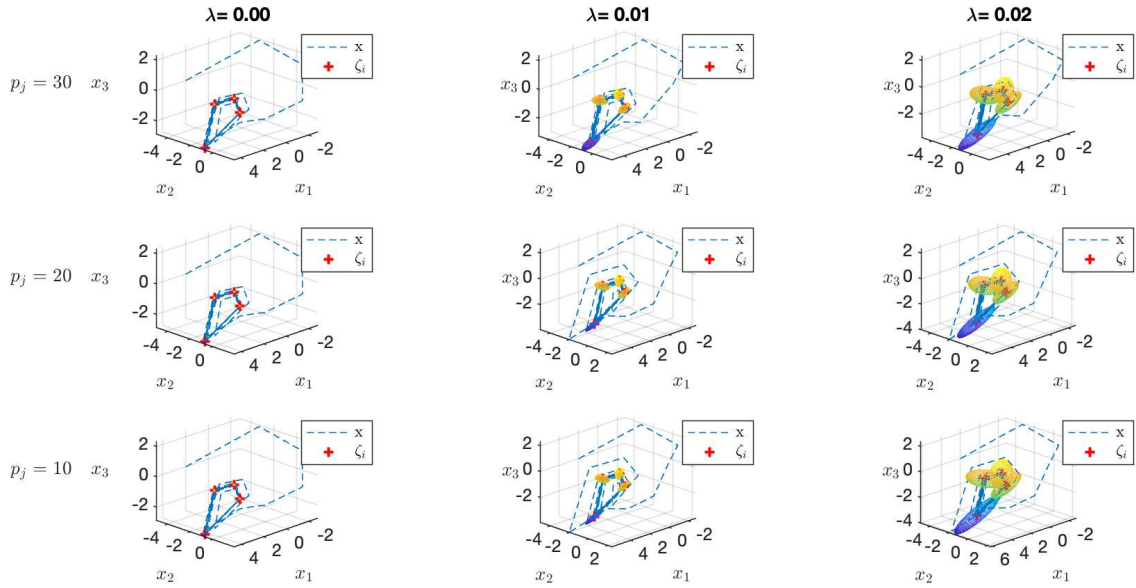}		  
\caption{Evolution of the state variables (blue dashed line) for cycle $\nu_2=\{1,2,2,2\}$ with their associated robust limit cycles. The red crosses represent $\zeta_i$, which are the centers of the ellipsoids $\mathcal E(W_i^{-1},\zeta_i)$. From top to bottom, the plot shows the attractors when the number of data available for the design are $30$, $20$ and $10$, respectively. From left to right, the uncertainty bound in the data increases $\lambda=0$, $0.01,$ and $0.02$.}
\label{fig:stab_to_limit_cycle2}
\end{figure*}

\begin{figure}
		\centering
		 \includegraphics[width=0.7\textwidth]{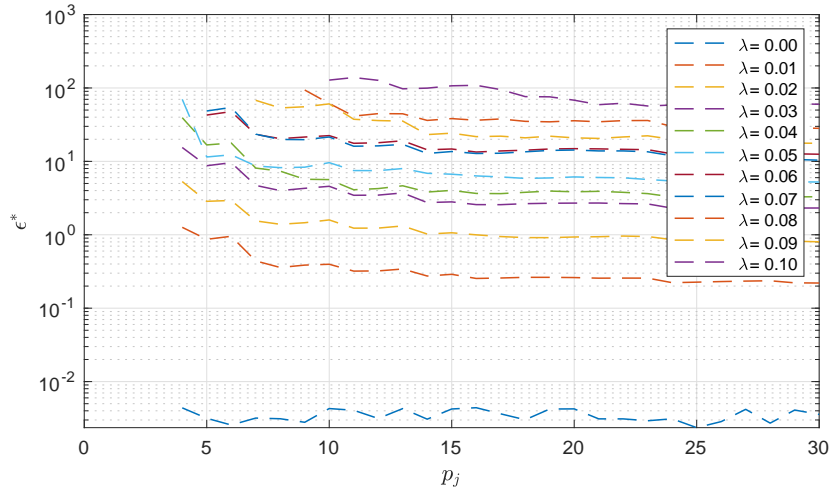}		  
		\caption{Graph showing the optimal value $\varepsilon^\ast$ of optimization problem \eqref{opt:data} for cycle $\nu=\{1,2\}$ and for various values of $\lambda$ from $0$ to $0.1$.}
		\label{fig:eps_p}
\end{figure}

\begin{table}[t]
\begin{center}
\begin{tabular}{|r|c||c|c|c||r|c|c|c|c|}
\multicolumn{2}{c}{Cycle:}& \multicolumn{3}{c}{ $\nu=\{1,2\}$}&\multicolumn{3}{c}{$\nu=\{1,2,2,2\}$}\\
\hline
$\varepsilon^*$  & Data
&$\lambda_j=0$& $\lambda_j=0.05$ &$\lambda_j=0.10$ &$\lambda_j=0$& $\lambda_j=0.01$ &$\lambda_j=0.02$\\
\hline
Th.\ref{th:theo_robust0} &$A_j,B_j$&  0.0005 &   1.1519   & 4.6078 &  0.0008   & 0.0709 &   0.2837\\

\hline

					&$p_j=30$ &  0.0063  &  2.4167 &  23.1562 &  0.0081  &  0.5751 &  2.8390\\  
Th.\ref{th:theo_data-robust} 	&$p_j=20$&    0.0065 &   2.8790 &  26.3315  &   0.0084  &  0.6905 &  3.1607\\
&$p_j=10$&    0.0090 &   3.8825 &  43.3567 &   0.0107  &  0.6922 &  4.5794\\ 
\hline
\end{tabular}
\end{center}
\caption{Optimal values of $\varepsilon^*$ obtained by solving problem~\eqref{opt:data} for various numbers of data $p_j=10$, $20$ and $30$, for cycle $\nu=\{1,2\}$.}
\label{Tab:eps}
\end{table}%


Finally, Figure~\ref{fig:eps_p} presents the variation of $\varepsilon^\ast$ with respect to the number of data measurements for various values of $\lambda$. First of all, this figure shows that optimization problem~\eqref{opt:data} has no solution when the number of data is less than $4=n+1$. This is consistent with the assumption of informativity of the data matrices $Y_j$ in Definition~\ref{def:informativity}. Moreover, as $\lambda$ increases, the same optimization problem requires more than the minimal necessary order of data ($n+1=4$) to obtain a solution. The minimum number of data to obtain a solution, increases with the amplitude of the disturbance. The general tendency is that the optimal solution decreases as the number of data increases. This figure also shows that the values of $\varepsilon^\ast$ tend to a constant as $p$ increases. 

The same example has been treated with a different cycle $\nu_2=\{1,2,2,2\}$. Applying again Theorem \ref{th:theo_robust0} and ~\ref{th:theo_data-robust} and solving the associated optimization problems, we have obtained the numerical results provided in Table~\ref{Tab:eps}. Note that the data-driven conditions deliver attractors composed of larger ellipsoids compared to the previous case with $\nu_1$. This explains why the amplitudes of the disturbance noise have been reduced to $\lambda=0,0.01$ and $0.02$. Note that the optimal value of $\varepsilon$ obtained by solving the model-based optimization problem~\eqref{opt:model} with $\lambda=0.1$ is $7.0$, which is already a very large value. Figure~\ref{fig:stab_to_limit_cycle2} presents the simulations resulting from the same initial condition as the one used in Figure~\ref{fig:model-based}.

In this example, we have compared the numerical results obtained for two cycles, i.e., one of short length, as $\nu_1$ and another one with a larger length, as $\nu_2$. Increasing the length of a cycle makes the attractor larger. Intuitively, this is because the uncertainties in the model matrices propagate from one element of the cycle to the next one, to guarantee the invariance of the attractor.

As a final comment on the numerical application of the data-driven design, it is worth noting that the numerical solutions highly depend on the noise and are very sensitive to it.

\section{Conclusion}\label{sec:conclusion}
In this paper, a first general result is presented that deals with the equivalence of LMIs, which can be seen as an alternative solution to the S-Lemma presented in \cite{vanWaarde2020noisy}, dedicated to the stabilization problem of linear systems. Then, to demonstrate the potential of this alternative, the objective of the paper is to develop a data-based stabilization criterion for the class of switched affine systems, which corresponds to a relevant class of nonlinear systems. First, a model-based control design for switched affine systems subject to a bounded uncertainty is presented as an extension of \cite{serieye2022attractor}. This method, based on the selection \textit{a priori} of a particular sequence of modes, provides a pure state-feedback control law that stabilizes the closed-loop systems to an attractor characterized by a level set of the Lyapunov function. The attractor is made up of the union of possibly disjoint ellipsoids resulting from an optimization problem. Then, this contribution is adapted to the data-driven control design, thanks to application of the first lemma.  An academic example validates the main result, showing that the desired attractor grows with larger noise and smaller amount of data. Future works would aim at applying this contribution to power converters modelled as switched affine systems.

\section*{Acknowledgments}
This work was supported in part by projects PID2019-105890RJ-I00 and PID2019-109071RB-I00, funded by MCIN/ AEI /10.13039/501100011033/  and FEDER A way of making Europe, by the Agence Nationale de la Recherche (ANR)-France under Grant ANR-18-CE40-0022-01.

\bibliography{bibalexv2}

\begin{thebibliography}{10}
\providecommand \doibase [0]{http://dx.doi.org/}%

\bibitem{ebihara2015s}
Ebihara Y, Peaucelle D, Arzelier D. {\it S-variable approach to {LMI}-based
  robust control}. 6.
\newblock Springer .
\newblock 2015.

\bibitem{postlethwaite2007robust}
Postlethwaite I, Turner M, Herrmann G. Robust control applications. {\it Annual
  Reviews in Control} 2007\string; 31(1)\string: 27--39.

\bibitem{scherer2001theory}
Scherer C. {\it Theory of robust control}.
\newblock Citeseer .
\newblock 2001.

\bibitem{boczar2018finite}
Boczar R, Matni N, Recht B. Finite-data performance guarantees for the
  output-feedback control of an unknown system. In: Proceedings of the IEEE
  Conference on Decision and Control. ; 2018\string: 2994--2999.

\bibitem{recht2019tour}
Recht B. A Tour of Reinforcement Learning: The View from Continuous Control.
  {\it Annual Review of Control, Robotics, and Autonomous Systems} 2019\string;
  2\string: 253--279.

\bibitem{berberich2020robust}
Berberich J, Koch A, Scherer C, Allg{\"o}wer F. Robust data-driven
  state-feedback design. In: Proceedings of the IEEE American Control
  Conference. ; 2020\string: 1532--1538.

\bibitem{berberich2020data}
Berberich J, K{\"o}hler J, M{\"u}ller MA, Allg{\"o}wer F. Data-driven model
  predictive control with stability and robustness guarantees. {\it IEEE Trans.
  on Automatic Control} 2020\string; 66(4)\string: 1702--1717.

\bibitem{hou2013model}
Hou Z, Wang Z. From model-based control to data-driven control: Survey,
  classification and perspective. {\it Information Sciences} 2013\string;
  235\string: 3--35.

\bibitem{dePersis2019formulas}
De~Persis C, Tesi P. Formulas for data-driven control: Stabilization,
  optimality, and robustness. {\it IEEE Trans. on Automatic Control}
  2019\string; 65(3)\string: 909--924.

\bibitem{vanWaarde2020noisy}
{Van Waarde} H, Camlibel M, Mesbahi M. From noisy data to feedback controllers:
  non-conservative design via a matrix {S}-lemma. {\it IEEE Trans. on Automatic
  Control} 2022\string; 67(1)\string: 162 - 175.

\bibitem{van2020data}
Van~Waarde H, Eising J, Trentelman H, Camlibel M. Data informativity: a new
  perspective on data-driven analysis and control. {\it IEEE Trans. on
  Automatic Control} 2020\string; 65(11)\string: 4753--4768.

\bibitem{berberich2020combining}
Berberich J, Scherer C, Allg{\"o}wer F. Combining prior knowledge and data for
  robust controller design. {\it arXiv preprint arXiv:2009.05253} 2020.

\bibitem{bisoffi2020data}
Bisoffi A, De~Persis C, Tesi P. Data-based stabilization of unknown bilinear
  systems with guaranteed basin of attraction. {\it Systems \& Control Letters}
  2020\string; 145\string: 104788.

\bibitem{mohler1973bilinear}
Mohler R. {\it Bilinear control processes: with applications to engineering,
  ecology, and medicine}. 106.
\newblock Academic Press .
\newblock 1973.

\bibitem{antunes2016linear}
Antunes D, Heemels W. Linear quadratic regulation of switched systems using
  informed policies. {\it IEEE Trans. on Automatic Control} 2016\string;
  62(6)\string: 2675--2688.

\bibitem{goebel2012hybrid}
Goebel R, Sanfelice R, Teel A. {\it Hybrid Dynamical Systems: modeling,
  stability, and robustness}.
\newblock Princeton University Press .
\newblock 2012.

\bibitem{liberzon2003switching}
Liberzon D. {\it Switching in systems and control}. 190.
\newblock Springer .
\newblock 2003.

\bibitem{egidio2020global}
Egidio L, Daiha H, Deaecto G. Global asymptotic stability of limit cycle and
  ${H}_2/{H}_\infty$ performance of discrete-time switched affine systems. {\it
  Automatica} 2020\string; 116\string: 108927.

\bibitem{serieye2020stabilization}
Serieye M, Albea C, Seuret A, Jungers M. Stabilization of switched affine
  systems via multiple shifted Lyapunov functions. {\it IFAC-PapersOnLine}
  2020\string; 53(2)\string: 6133--6138.

\bibitem{kenanian2019data}
Kenanian J, Balkan R, Tabuada P. Data driven stability analysis of black-box
  switched linear systems. {\it Automatica} 2019\string; 109\string: 108533.

\bibitem{rotulo2021online}
Rotulo M, De~Persis C, Tesi P. Online data-driven stabilization of switched
  linear systems. In: Proceedings of the European Control Conference. ;
  2021\string: 300--305.

\bibitem{della2021data}
Della~Rossa M, Wang Z, Egidio L, Jungers R. Data-driven stability analysis of
  switched affine systems. {\it arXiv preprint arXiv:2109.11169} 2021.

\bibitem{smarra2018data}
Smarra F, Jain A, Mangharam R, D’Innocenzo A. Data-driven switched affine
  modeling for model predictive control. {\it IFAC-PapersOnLine} 2018\string;
  51(16)\string: 199--204.

\bibitem{serieye2022attractor}
Serieye M, Albea C, Seuret A, Jungers M. Attractors and limit cycles of
  discrete-time switching affine systems: nominal and uncertain cases. {\it
  Automatica} 2023\string; 149\string: 110691.

\bibitem{Strogatz_Book_1994}
Strogatz S. {\it Nonlinear Dynamics and Chaos: With Applications to Physics,
  Biology, Chemistry and Engineering}.
\newblock Perseus Books .
\newblock 1994.

\bibitem{Sun_CSF_2008}
Sun Y. Existence and uniqueness of limit cycle for a class of nonlinear
  discrete-time systems. {\it Chaos, Solitons \& Fractals} 2008\string;
  38(1)\string: 89--96.

\bibitem{deaecto2016stability}
Deaecto G, Geromel J. Stability analysis and control design of discrete-time
  switched affine systems. {\it IEEE Trans. on Automatic Control} 2016\string;
  62(8)\string: 4058--4065.

\end{thebibliography}
\clearpage


\end{document}